%% file: paper.tex
\documentclass[acmsmall,screen]{acmart}

\setcopyright{rightsretained}
\acmPrice{}
\acmDOI{10.1145/3571254}
\acmYear{2023}
\copyrightyear{2023}
\acmSubmissionID{popl23main-p528-p}
\acmJournal{PACMPL}
\acmVolume{7}
\acmNumber{POPL}
\acmArticle{61}
\acmMonth{1}
\received{2022-07-07}
\received[accepted]{2022-11-07}


\usepackage{stmaryrd}
\usepackage{amsmath}
\usepackage{xspace}
\usepackage{syntax}
\usepackage{semantic}
\usepackage{mathpartir}
\usepackage{wrapfig}
\usepackage{tikz}
\usetikzlibrary{arrows, automata, positioning}
\usepackage{forest}
\useforestlibrary{linguistics}
\usepackage{subcaption}
\usepackage{mathtools}
\usepackage{combelow}
\usepackage{orcidlink}
\definecolor{cerulean}{rgb}{0.0, 0.48, 0.65}
\definecolor{chromeyellow}{rgb}{1.0, 0.65, 0.0}

\usepackage{listings,lstlangcoq}
\lstdefinestyle{customcoq}{
  columns=flexible,
  mathescape=true,
  belowcaptionskip=1\baselineskip,
  breaklines=true,
  xleftmargin=\parindent,
  language=Coq,
  morekeywords={Variant, fun, Arguments, Type, cofix},
  emph={%
    SOCKAPI,ITree,data_at,data_at_
  },
  emphstyle={\bfseries\color{green!40!red!80}},
  showstringspaces=false,
  basicstyle=\small\ttfamily,
  keywordstyle=\bfseries\color{green!20!black},
  commentstyle=\itshape\color{red!40!black},
  identifierstyle=\color{violet!50!black},
  stringstyle=\color{orange},
  escapeinside={<@}{@>}
}
\newcommand{\inlinecoq}[1]{\mbox{\lstinline[style=customcoq,columns=fixed,basewidth=.48em]{#1}}}
\newcommand{\ilc}[1]{\inlinecoq{#1}}

\newif\ifarxiv\arxivtrue   
\ifarxiv
  \newcommand{\arxiv}[2]{#1}
\else
  \newcommand{\arxiv}[2]{#2}
\fi
\newif\ifcomments\commentstrue   
\newif\ifaftersubmission \aftersubmissionfalse 
\newif\ifplentyofspace \plentyofspacefalse 
\ifcomments
\newcommand{\proposecut}[1]{\ifcomments{\color{gray} #1}\fi}
\newcommand{\yz}[1]{\textcolor{blue}{{[YZ:~#1]}}}
\newcommand{\sz}[1]{\textcolor{brown!100!black!100}{{[SZ:~#1]}}}
\newcommand{\ph}[1]{\textcolor{olive}{{[PH:~#1]}}}
\newcommand{\lh}[1]{\textcolor{green!75!black}{{[LH:~#1]}}}
\newcommand{\nc}[1]{\textcolor{orange}{{[NC:~#1]}}}
\else
\newcommand{\proposecut}[1]{}
\newcommand{\yz}[1]{}
\newcommand{\sz}[1]{}
\newcommand{\ph}[1]{}
\newcommand{\lh}[1]{}
\newcommand{\nc}[1]{}
\fi

\include{macros}

\begin{document}

\title{Choice Trees}
\subtitle{Representing Nondeterministic, Recursive, and Impure Programs in Coq}


\author{Nicolas Chappe}
\orcid{0000-0003-3732-7704}
\affiliation{
  \institution{Univ Lyon, EnsL, UCBL, CNRS, Inria,  LIP, F-69342, LYON Cedex 07}
  \country{France}
}
\email{nicolas.chappe@ens-lyon.fr}

\author{Paul He}
\orcid{0000-0002-6305-4335}
\affiliation{
  \institution{University of Pennsylvania}
  \city{Philadelphia}
  \state{PA}
  \country{USA}
}
\email{paulhe@cis.upenn.edu}

\author{Ludovic Henrio}
\orcid{0000-0001-7137-3523}
\affiliation{
  \institution{Univ Lyon, EnsL, UCBL, CNRS, Inria,  LIP, F-69342, LYON Cedex 07}
  \city{Lyon}
  \country{France}
}
\email{ludovic.henrio@cnrs.fr}

\author{Yannick Zakowski}
\orcid{0000-0003-4585-6470}
\affiliation{
  \institution{Univ Lyon, EnsL, UCBL, CNRS, Inria,  LIP, F-69342, LYON Cedex 07}
  \country{France}
}
\email{yannick.zakowski@inria.fr}

\author{Steve Zdancewic}
\orcid{0000-0002-3516-1512}
\affiliation{
  \institution{University of Pennsylvania}
  \country{USA}
}
\email{stevez@cis.upenn.edu}


\begin{abstract}
  This paper introduces Choice Trees (\ctreesn), a monad for modeling nondeterministic, recursive, and
  impure programs in \coq.
  Inspired by \citeauthor{XZHH+20}'s \itreesn, this novel data structure embeds computations into
  coinductive trees with three kind of nodes: external events, and two
  variants of nondeterministic branching.
  This apparent redundancy allows us to provide shallow embedding of
  denotational models with internal choice in the style of \ccs, while
  recovering an inductive LTS view of the computation.
  \ctreesn inherit a vast collection of bisimulation and refinement tools, with respect to
  which we establish a rich equational theory.

  We connect \ctreesn to the \itreesn infrastructure by showing how a monad
  morphism embedding the former into the latter permits to use \ctreesn to
  implement nondeterministic effects.
  We demonstrate the utility of \ctreesn by using them to model
  concurrency semantics in two case studies: \ccs and cooperative multithreading.
\end{abstract}


\maketitle



\section{Introduction}
\label{sec:intro}

\input{introduction}

\section{Ctrees: Definition and Combinators}
\label{sec:ctrees}
\input{core}

\section{Equivalences and Equational Theory for \ctreesn}
\label{sec:bisim}
\input{bisim}

\section{Interpretation from and to \ctreesn}
\label{sec:interp}
\input{interp}

\section{Case study: a model for ccs}
\label{sec:ccs}
\input{ccs}

\section{Case study: modeling cooperative multithreading}
\label{sec:yield}
\input{yield}

\section{Related Work}
\label{sec:rw}
\input{rw}
\section{Conclusion and Perspectives}
\label{sec:conclusion}
\input{conclusion}

\subsection*{Acknowledgements}
We are grateful to the anonymous reviewers for their in-depth comments that
helped both improve this paper, and open avenues of further work.
We thank Gabriel Radanne for providing assistance with TikZ.
Finally, we are most thankful to Damien Pous for developing the coinduction
library this formal development crucially relies on, and for the numerous
advice he provided us with.

\arxiv{
\appendix

\input{appendix}
}{}

\bibliographystyle{ACM-Reference-Format}
\bibliography{references}

\end{document}

%% file: macros.tex
\newcommand{\pfpn}{post fixpoint\xspace}

\newcommand{\gfpn}{greatest fixpoint\xspace}

\newcommand{\cofixes}{cofixes\xspace}

\newcommand{\gfpc}[1]{{\color{cerulean}{#1}}}
\newcommand{\gfptop}[1]{\ensuremath{\gfpc{\tm{gfp}~#1}}}
\newcommand{\gfp}[1]{\ensuremath{\gfpc{\tm{gfp}~\lambda #1}~\cdot}}
\newcommand{\cofix}[1]{\ensuremath{\gfpc{\tm{cofix}~#1}}}
\newcommand{\compa}[1]{\ensuremath{\tm{t}_{#1}}}

\newcommand{\equn}{\tm{equ}\xspace}

\newcommand{\sembra}[1]{\ensuremath{\llbracket #1 \rrbracket}}
\newcommand{\ssembra}[1]{\mathcal{S}\ensuremath{\llbracket #1 \rrbracket}}
\newcommand{\lts}{LTS}
\newcommand{\equ}{\ensuremath{\cong}}

\newcommand{\myurl}{https://github.com/vellvm/ctrees/blob/popl23}
\newcommand{\linksrc}{\href{https://github.com/vellvm/ctrees}{\includegraphics[width=0.3cm,height=0.3cm,keepaspectratio]{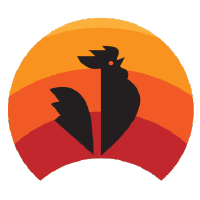}}}
\newcommand{\link}[3]{\href{#1/#2/#3}{\includegraphics[width=0.3cm,height=0.3cm,keepaspectratio]{logo.png}}}
\newcommand{\linkt}[2]{\link{\myurl}{theories}{#1\#L#2}}
\newcommand{\linkccs}[2]{\link{\myurl}{examples/CCS}{#1\#L#2}}
\newcommand{\linkimp}[2]{\link{\myurl}{examples/ImpBr}{#1\#L#2}}
\newcommand{\linkyield}[2]{\link{\myurl}{examples/Yield}{#1\#L#2}}

\newcommand{\mkw}[1]{\mbox{\sc{#1}}}
\newcommand{\reflC}{\ensuremath{\mkw{refl}^{up}}\xspace}
\newcommand{\symC}{\ensuremath{\mkw{sym}^{up}}\xspace}
\newcommand{\transC}{\ensuremath{\mkw{trans}^{up}}\xspace}
\newcommand{\visC}{\ensuremath{\mkw{Vis}^{up}}\xspace}
\newcommand{\bindC}[1]{\ensuremath{\mkw{bind}^{up}(#1)}\xspace}
\newcommand{\uptorelC}[1]{\ensuremath{\mkw{upto}^{up}(#1)}\xspace}

\newcommand{\itreen}{ITree\xspace}
\newcommand{\itreesn}{ITrees\xspace}
\newcommand{\delay}{\tm{later}\xspace}
\newcommand{\itau}{\tm{Tau}\xspace}

\newcommand{\ctreen}{CTree\xspace}
\newcommand{\ctreesn}{CTrees\xspace}
\newcommand{\visn}{\tm{Vis}\xspace}
\newcommand{\brSn}{\ensuremath{Br_S}\xspace}
\newcommand{\brDn}{\ensuremath{Br_D}\xspace}
\newcommand{\Vis}[2]{\ensuremath{Vis~#1~#2}}
\newcommand{\brS}[2]{\ensuremath{Br_S^#1~#2}}
\newcommand{\brStwo}[2]{\ensuremath{Br_S^2~#1~#2}}
\newcommand{\brD}[2]{\ensuremath{Br_D^#1~#2}}
\newcommand{\brDtwo}[2]{\ensuremath{Br_D^2~#1~#2}}
\newcommand{\brDthree}[3]{\ensuremath{Br_D^3~#1~#2~#3}}
\newcommand{\brDfour}{\ensuremath{Br_D^4}}
\newcommand{\Ret}[1]{\ensuremath{\tm{Ret}~#1}}
\newcommand{\ret}[1]{\ensuremath{\tm{ret}~#1}}
\newcommand{\bind}{~\tm{;\!;}~}
\newcommand{\bindc}{>\!\!>\!\!=}
\newcommand{\bindn}{\ilc{bind}}
\newcommand{\bget}{\gets}

\newcommand{\trigger}[1]{\ensuremath{\tm{trigger}~(#1)}}
\newcommand{\triggernoparens}[1]{\ensuremath{\tm{trigger}~#1}}

\newcommand{\interp}{\tm{interp}}

\newcommand{\gethead}[1]{\tm{head~#1}}
\newcommand{\guard}{\tm{Guard}\xspace}
\newcommand{\step}{\tm{Step}\xspace}
\newcommand{\isstuck}[1]{#1\not\rightarrow}

\newcommand{\stuck}{\ensuremath{\emptyset}\xspace}
\newcommand{\spinDn}{\ensuremath{\tm{spinD\_nary}}\xspace}
\newcommand{\spinSn}{\ensuremath{\tm{spinS\_nary}}\xspace}

\newcommand{\lretKW}{\tm{val}}
\newcommand{\lobsKW}{\tm{obs}}
\newcommand{\llabel}{\tm{label}\xspace}
\newcommand{\lret}[1]{\lretKW~#1\xspace}
\newcommand{\lobs}[2]{\lobsKW~#1~#2\xspace}
\newcommand{\ltau}{\tm{tau}\xspace}
\newcommand{\lstep}[3]{#1 \xrightarrow{#2} #3\xspace}
\newcommand{\wstep}[3]{#1 \xRightarrow{#2} #3\xspace}

\tikzset{
  treenode/.style =
  {align=center, inner sep=0pt, text centered, font=\sffamily},
  ncv/.style =
  {treenode, circle, gray, font=\sffamily\bfseries, draw=gray, fill=gray,
    text width=1.5em},
  nci/.style = {treenode, circle, white, draw=gray,
    text width=1.5em, very thick},
  nv/.style = {treenode, circle, draw=black,
    minimum width=0.5em, minimum height=0.5em},
  nr/.style = {treenode, rectangle, draw=black,
    minimum width=0.5em, minimum height=0.5em},
  ntau/.style =
  {treenode, circle, gray, font=\sffamily\bfseries, draw=gray, fill=gray,
    text width=1.5em}

}

\newcommand{\true}{\nt{true}}
\newcommand{\false}{\nt{false}}
\newcommand{\imp}{\tm{imp}\xspace}
\newcommand{\ccomm}{\tm{comm}\xspace}
\newcommand{\cskip}{\tm{skip}\xspace}
\newcommand{\cassign}[2]{#1 \tm{::=} ~#2\xspace}
\newcommand{\cseq}[2]{#1 {;} ~#2\xspace}
\newcommand{\cwhile}[2]{\tm{while} ~#1 ~\tm{do} ~#2\xspace}
\newcommand{\cbr}[2]{\tm{br} ~#1 ~\tm{or} ~#2\xspace}
\newcommand{\cbrn}{\tm{br}\xspace}
\newcommand{\cstuck}{\tm{block}\xspace}
\newcommand{\cprint}{\tm{print}\xspace}
\newcommand{\cfork}[2]{\tm{fork} ~#1~#2\xspace}
\newcommand{\cyield}{\tm{yield}\xspace}

\newcommand{\ebr}{\tm{flip}\xspace}
\newcommand{\efork}{\tm{fork}\xspace}

\newcommand{\ccs}{\tm{ccs}\xspace}
\newcommand{\ccssem}{\ensuremath{\tm{ccs}^\#}\xspace}
\newcommand{\pP}{\nt{P}}
\newcommand{\pQ}{\nt{Q}}

\newcommand{\comm}{\nt{c}}
\newcommand{\pnil}{0}
\newcommand{\prf}[2]{#1\cdot #2}
\newcommand{\pls}[2]{#1\oplus #2}
\newcommand{\para}[2]{#1\parallel #2}
\newcommand{\compat}[2]{#1\not \in #2}
\newcommand{\new}[2]{\nu #1\cdot #2}
\newcommand{\bang}[1]{! #1}
\newcommand{\stepccs}[3]{#1 \xrightarrow{#2}_{\scriptsize{\mathrm{ccs}}} #3}

\newcommand{\spnil}{\overline{0}}
\newcommand{\sprf}[2]{#1\overline{\cdot} #2}
\newcommand{\spls}[2]{#1\overline{\oplus} #2}
\newcommand{\spara}[2]{#1\overline{\parallel} #2}
\newcommand{\snew}[2]{\nu #1\overline{\cdot} #2}
\newcommand{\sparabang}[2]{#1\overline{\parallel!} #2}
\newcommand{\sbang}[1]{\overline{!} #1}

\newcommand{\actL}{\tm{actL}}
\newcommand{\actR}{\tm{actR}}
\newcommand{\actLR}{\tm{actLR}}
\newcommand{\actRPB}{\tm{pbR}}
\newcommand{\actRRPB}{\tm{pbRR}}
\newcommand{\actLRPB}{\tm{pbLR}}

\newcommand{\nt}[1]{\ensuremath{\mathit{#1}}} 
\newcommand{\tm}[1]{\ensuremath{\mathtt{#1}}} 

\newcommand{\PROP}{\ensuremath{\mathbb{P}}}

\newcommand{\sep}{~\mid~}

\newcommand{\defeq}{\triangleq}

\newcommand{\ir}{LLVM IR\xspace}

\newcommand{\coq}{\tm{Coq}\xspace}
\newcommand{\ocaml}{\tm{OCaml}\xspace}
\newcommand{\impbr}{\tm{ImpBr}\xspace}

\newcommand{\unit}{\ensuremath{\mathtt{(\:)}}}

\newcommand{\sem}[1]{\ensuremath{\llbracket #1 \rrbracket}}

\newcommand{\euttn}{\tm{eutt}\xspace}

\newcommand{\schedule}{\ensuremath{\mathsf{schedule}}\xspace}

\newcommand{\sbisim}{\sim}
\newcommand{\sbisimR}{\sim_R}
\newcommand{\sbisimn}{\ensuremath{\tm{sbisim}}\xspace}
\newcommand{\wbisim}{\approx}

\newcommand{\treq}{\equiv_{\tm{tr}}}

\newcommand{\ssim}{\lesssim}


%% file: introduction.tex
Reasoning about and modeling nondeterministic computations is important for many
purposes. Formal specifications use nondeterminism to abstract away from the
details of implementation choices.  Accounting for nondeterminism is crucial
when reasoning about the semantics of concurrent and distributed systems, which
are, by nature, nondeterministic due to races between threads, locks, or message
deliveries.  Consequently, precisely defining nondeterministic behaviors and
developing the mathematical tools to work with those definitions has been an
important research endeavor, and has led to the development of
formalisms like nondeterministic automata, labeled transition systems
and relational operational semantics~\cite{BPS01}, powerdomains~\cite{Smy76}, or
game semantics~\cite{AM99,RW11}, among others, all of which have been used to give semantics to
nondeterministic programming language features such as
concurrency~\cite{SW01,ccs,Harper13}.

In this paper, we are interested in developing tools for modeling
nondeterministic computations in a dependent type theory such as Coq's
CIC~\cite{coq}.  Although any of the formalisms mentioned above could be used
for such purposes (and many have been~\cite{compcerttso,promising,promising2,KS20,VMSJ+22}), those
techniques offer various tradeoffs when it comes to the needs of formalization:
automata, and labeled transitions systems, while offering powerful
bisimulation proof principles, are not easily made modular (except, perhaps, with complex extensions to the framework~\cite{henrio:01299562}).  Relationally-defined operational semantics are flexible and
expressive, but again suffer from issues of compositionality, which
makes it challenging to build general-purpose libraries
that support constructing complex models.  Conversely, powerdomains and game
semantics are more denotational approaches, aiming to ensure compositionality
by construction; however, the mathematical structures involved are themselves
very complex, typically involving many relations and constraints~\cite{AM99,MM07,RW11} that are not easy to implement in constructive logic (though there are some notable exceptions~\cite{KS20,VMSJ+22}).
Moreover, in all of the above-mentioned approaches, there are other tensions at
play. For instance, how ``deep'' the embedding is affects the amount of effort
needed to implement a formal semantics---``shallower'' embeddings typically
allow more re-use of metalanguage features, e.g., meta-level function
application can obviate the need to define and prove properties about a
substitution operation; ``deeper'' embeddings can side-step meta-level
limitations (such as Coq's insistence on pure, total functions) at the cost of
additional work to define the semantics.  Moreover, these tradeoffs can
have significant impact on how difficult it is to use other tools and
methodologies: for instance, to use QuickChick~\cite{quickchick}, one must be able to extract an
executable interpreter from the semantics, something that isn't always easy or
possible.

This paper introduces a new formalism designed specifically to facilitate the
definition of and reasoning about nondeterministic computations in Coq's
dependent type theory.  The key idea is to update Xia, et al.'s \textit{interaction trees} (\itreesn)
framework \cite{XZHH+20} with native support for nondeterminisic ``choice
nodes'' that represent internal choices made during computation. The
main technical contributions of this paper are to introduce the definition of
these \ctreesn{} (``choice trees'') and to develop the suitable metatheory and equational reasoning principles to accommodate that change.

We believe that \ctreesn{} offer an appealing, and novel, point in the design
space of formalisms for working with nondeterministic specifications within type
theory.  Unlike purely relational specifications, \ctreesn build nondeterminism
explicitly into a datatype, as nodes in a tree, and the nondeterminism is
realized propositionally at the level of the equational theory, which determines
when two \ctreen{} computations are in bisimulation.  This means that the user
of \ctreesn{} has more control over how to represent nondeterminism and when to apply the
incumbent propositional reasoning. By reifying the choice construct into a
data structure, one can write meta-level functions that manipulate \ctreesn{},
rather than working entirely within a relation on syntax.  This design allows us
to bring to bear the machinery of monadic interpreters to refine the
nondeterminism into an (executable) implementation.
While \itreesn{}
can represent such choice nodes, in our experience, using that feature to model
``internal'' nondeterminism is awkward: the natural equational theory for
\itreesn{} is too fine, and other techniques, such as interpretation into
\ilc{Prop}, don't work out neatly.

At the same time, the notion of bisimulation for \ctreesn{} is still connected
to familiar definitions like those from labeled transition systems (\lts), meaning
that much of the well-developed theory from prior work can be imported whole-sale.
Indeed, we define bisimilarity for \ctreesn{} by viewing them as LTSs and
applying standard definitions.  The fact that the definition of bisimulation
ends up being subtle and nontrivial is a sign that we gain something by working
with the \ctreesn{}: like their \itreen{} predecessors, \ctreesn{} have
compositional reasoning principles, the type is a monad, and the useful
combinators for working with \itreesn{}, namely sequential composition,
iteration and recursion, interpretation, etc., all carry over directly.
\ctreesn{}, though, further allow us to
conveniently, and flexibly, define nondeterminstic semantics, ranging from
simple choice operators to various flavors of parallel composition.  The benefit
is that, rather than just working with a ``raw'' LTS directly, we can construct
one using the \ctreesn{} combinators---this is a big benefit because, in
practice, the LTS defining the intended semantics of a nondeterministic
programming language cannot easily be built in a compositional way without using some kind of intermediate representation, which is exactly what \ctreesn provides (see the discussion
about Figure~\ref{fig:choices}).  A key technical novelty of our \ctreesn{}
definition is that it makes a distinction between \textit{stepping} choices (which correspond to $\tau$ transitions and introduce new LTS states) and \textit{delayed} choices (which don't correspond to a transition and don't create a state in the LTS). This design allows for compositional
construction of the LTS and generic reasoning rules that are usable in any
context.

The net result of our contributions is a library, entirely formalized in Coq,
that  offers flexible building blocks for constructing nondeterministic, and
hence concurrent, models of computation.  To demonstrate the applicability of this library, we
use it to implement the semantics from two different formalisms: \ccs~\cite{Milner:CC1989} and a language with
cooperative threads inspired from the literature~\cite{Abadi2010}.  Crucially, in both of these scenarios, we are
able to define the appropriate parallel composition combinators such
that the semantics of the programming language can be defined fully
compositionally (i.e., by straightforward induction on the syntax).  Moreover, we
recover the classic definition of program equivalence for \ccs directly from the
equational theory induced by the encoding of the semantics using \ctreesn{}; for the language with cooperative threading, we prove some standard program equivalences.

To summarize, this paper makes the following contributions:

\begin{itemize}
\item We introduce \ctreesn{}, a novel data structure for defining
  nondeterministic computations in type theory, along with a set of combinators
  for building semantic objects using \ctreesn{}.
\item We develop the appropriate metatheory needed to reason about strong and
  weak bisimilarity of \ctreesn{}, connecting their semantics to concepts
  familiar from labeled transition systems.
\item We show that \ctreesn{} admit appropriate notions of refinement and that
  we can use them to construct monadic interpreters; we show that \itreesn{} can
  be faithfully embedded into \ctreesn{}.
\item We demonstrate how to use \ctreesn{} in two case studies: (1) to define a
  semantics for Milner's classic \ccs and prove that the resulting derived
  equational theory coincides with the one given by the standard operational
  semantics, and (2) to model cooperative multithreading with support for
  \efork and \cyield operations and prove nontrivial program equivalences.
\end{itemize}

All of our results have been implemented in \coq, and all claims in this paper
are fully mechanically verified.
For expository purposes, we stray away from \coq's syntax in the body of this
paper, but systematically link our claims to their formal counterpart via
hyperlinks represented as (\linksrc).

The remainder of the paper is organized as follows. The next section gives some
background about interaction trees and monadic interpreters, along with a
discussion of the challenges of modeling nondeterminism, laying the foundation
for our results. We introduce the \ctreesn{} data structure and its main
combinators in Section \ref{sec:ctrees}. Section~\ref{sec:bisim} introduces
several notions of equivalences over \ctreesn---(coinductive) equality, strong
bisimilarity, weak bisimilarity, and trace equivalence---and describes its core
equational theory. Section~\ref{sec:interp} describes how to interpret
uninterpreted events in an \itreen into ``choice'' branches in a \ctreen, as
well as how to define the monadic interpretation of events from \ctreesn.
Section~\ref{sec:ccs} describes our first case study, a model for \ccs.
Section~\ref{sec:yield} describes our second case study, a model for the \imp
language extended with cooperative multithreading. Finally, Section~\ref{sec:rw}
discusses related work and concludes.

\section{Background}
\label{sec:background}

\subsection{Interaction trees and monadic interpreters}
\label{sec:itrees}

Monadic interpreters have grown to be an attractive way to mechanize the
semantics of a wide class of computational systems in dependent typed theory,
such as the one found in many proof assistants, for which the host language is
purely functional and total.  In the Coq ecosystem, interaction
trees~\cite{XZHH+20} provide a rich library for building and reasoning about
such monadic interpreters.  By building upon the free(r)
monad~\cite{freer,freespec}, one can both design highly reusable components, as
well as define modular models of programming languages more amenable to
evolution.  By modeling recursion coinductively, in the style of Capretta's
delay monad~\cite{Cap05,ADK17}, such interpreters can model non-total object
languages while retaining the ability to \emph{extract} correct-by-construction,
executable, reference interpreters.  By generically lifting monadic
implementations of effects into a monad homomorphism, complex interpreters can
be built by stages, starting from an initial structure where all effects are
free and incrementally introducing their implementation.  Working in a proof
assistant, these structures are well suited for reasoning about program
equivalence and program refinement: each monadic structure comes with its own
notion of refinement, and the layered infrastructure gives rise to increasingly
richer equivalences~\cite{YZZ22}, starting from the free monad, which comes
with no associated algebra.

\begin{figure}

\begin{lstlisting}[style=customcoq,basicstyle=\small\ttfamily]
  CoInductive itree (E: Type -> Type) (R: Type) : Type :=
  | Ret (r: R)                                  (* computation terminating with value r *)
  | Vis {A: Type} (e : E A) (k : A -> itree E R) (* event e yielding an answer in A *)
  | later (t: itree E R).                       (* "silent" tau transition with child t *)

\end{lstlisting}

\vspace{-1ex}
\caption{Interaction trees: definition}
\label{fig:itree-def}
\end{figure}

Interaction trees are coinductive data structures for representing (potentially
divergent) computations that interact with an external environment through
\textit{visible events}.  A definition of the \itreen datatype is shown in
Figure~\ref{fig:itree-def}.\footnote{The signature of \itreesn is presented with
  a positive coinductive datatype for expository purposes. The actual
  implementation is defined in the negative style.}  The datatype takes as its
first parameter a signature---described as a family of types \ilc{E : Type ->
  Type}---that specifies the set of interactions the computation may have with
the environment.  The \ilc{Vis} constructor builds a node in the tree
representing such an interaction, followed by a continuation indexed by the
return type of the event.  The second parameter, \ilc{R}, is the \textit{result
  type}, the type of values that the entire computation may return, if it halts. The
constructor \ilc{Ret} builds such a pure computation, represented as a leaf.
Finally, the \delay constructor models an internal, non-observable step of
computation, allowing the representation of silently diverging computations; it
is also used for guarding corecursive definitions.\footnote{The \itreen library
  uses \ilc{Tau} to represent \delay nodes.  \ilc{Tau} and $\tau$ are overloaded
  in our context, so we rename it to \delay here to avoid ambiguity. \ctreesn
  will replace the \delay (i.e. \ilc{Tau}) constructor with a more general
  construct anyway.}

To illustrate the approach supported by \itreesn, and motivate the contributions
of this paper, we consider how to define the semantics for a simple imperative programming language, \imp:
\begin{mathpar}
  \ccomm \defeq \cskip \sep \cassign{x}{e} \sep \cseq{c1}{c2} \sep \cwhile{b}{c}
\end{mathpar}
The language contains a \cskip construct, assignments, sequential composition, and loops---we
assume a simple language of expressions, $e$, that we omit here.
Consider the following \imp programs:
\begin{mathpar}
  p_1 \defeq \cwhile{\true}{\cskip} \and
  p_2 \defeq \cseq{\cassign{x}{0}}{\cassign{x}{y}} \and
  p_3 \defeq \cassign{x}{y}
\end{mathpar}

Following a semantic model for \imp built on \itreesn in the style of
\citeauthor{XZHH+20}~\cite{XZHH+20}, one builds a semantics in two stages. First,
commands are represented as monadic computations of type \ilc{itree MemE unit}:
commands do not return values, so the return type of the computation is
the trivial \ilc{unit} type; interactions with the memory are (at first) left
uninterpreted, as indicated by the event signature \ilc{MemE}.
This signature encodes two operations:  \ilc{rd} yields a
value, while  \ilc{wr} yields only the acknowledgment that the operation took
place, which we encode again using \ilc{unit}.

\begin{lstlisting}[style=customcoq]
Variant MemE : Type -> Type :=
  | rd (x : var)                : MemE value
  | wr (x : var) (v : value) : MemE unit
\end{lstlisting}

Indexing by the \ilc{value} type in the continuation of \ilc{rd} events
gives rise to non-unary branches in the tree representing these programs. For
instance, the programs $p_1,p_2,p_3$ are, respectively, modeled at this stage by
the trees shown in Figure~\ref{fig:trees}. These diagrams omit the \ilc{Vis} and
\ilc{Ret} constructors, because their presence is clear from the picture. For
example, the second tree would be written  as
\[p_2 = \ilc{Vis (wr x 0) (fun _ => Vis (rd y) (fun ans => (Vis (wr x ans) (fun _ => Ret tt))))}.\]

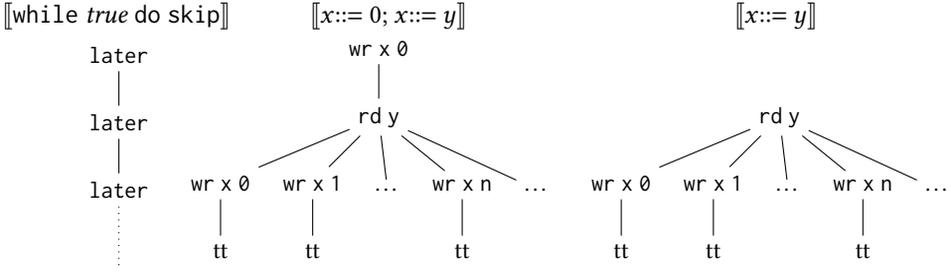
\begin{figure}
  \[
    \begin{array}{ccc}
      \multicolumn{3}{l}{\sem{\cwhile{\true}{\cskip}} \qquad \quad \sem{\cseq{\cassign{x}{0}}{\cassign{x}{y}}} \hspace{1.4in} \sem{\cassign{x}{y}} } \\
      \qquad \quad
{\small
 \begin{forest}
  [\delay
  [\delay
  [\delay, name=foo
  ]]]
  \draw[dotted] (foo) --+(0,-1);
\end{forest}}
      &
{\small
\begin{forest}
  [\tm{wr~x~0}
  [\tm{rd~y}
  [\tm{wr~x~0}[tt]] [\tm{wr~x~1}[tt]] [\dots] [\tm{wr~x~n}[tt]] [\dots]
  ]]
\end{forest}}
      &
{\small
\begin{forest}
  [\tm{rd~y}
  [\tm{wr~x~0}[tt]] [\tm{wr~x~1}[tt]] [\dots] [\tm{wr~x~n}[tt]] [\dots]
  ]
\end{forest}}
\end{array}
\]
\caption{Example \itreesn denoting the \imp programs $p_1$, $p_2$, and $p_3$.}
\label{fig:trees}
\end{figure}

The \delay nodes in the first tree are the guards from
Capretta's monad: because the computation diverges silently, it is modeled as
an infinite sequence of such guards. The equivalence used for
computations in the \itreen monad is a \emph{weak bisimulation}, dubbed
\emph{equivalence up-to taus} (\euttn), which allows one to ignore finite sequences of
\delay nodes when comparing two trees. It remains termination-sensitive: the silently diverging
computation is not equivalent to any other \itreen.

With \itreesn, no assumptions about the semantics of the uninterpreted memory
events is made.  Although one would expect $p_2$ and $p_3$ to be equivalent
as \imp programs, their trees are not \euttn since the former starts with a different
event than the latter.  This missing algebraic equivalence is concretely
recovered at the second stage of modeling: \imp programs are given a semantics
by interpreting the trees into the state monad, by \emph{handling} the
\ilc{MemE} events.  This yields computations in \ilc{stateT mem (itree voidE)
  unit}, or, unfolding the definition, \ilc{mem -> itree voidE (mem * unit)}.  Here, \ilc{voidE} is the
``empty'' event signature, such that an \itreen at that type either silently diverges
or deterministically returns an answer.  For $p_2$ and $p_3$, assuming an initial state
\ilc{m}, the computations become (writing the trees horizontally to save space):
\[
  \begin{array}{lcl}
\ilc{interp}\ h_{mem}\ \sem{p_2} ~\ilc{m} &=&
    \delay - \delay - \delay - (\ilc{m}\{x\gets 0\}\{x\gets \ilc{m}(y)\}, tt)
\\
\ilc{interp}\ h_{mem}\ \sem{p_3}~ \ilc{m} &=&
    \delay - \delay - (\ilc{m}\{x\gets \ilc{m}(y)\}, tt)
  \end{array}
\]


The \delay nodes are introduced by the interpretation of the memory events. More precisely,  an \ilc{interp} combinator  applies
the handler $h_{mem}$ to the \ilc{rd} and \ilc{wr} nodes of the trees, implementing their
semantics in terms of the state monad.  Assuming an appropriate implementation
of the memory, one can show that \(\ilc{m}\{x\gets0\}\{x\gets \ilc{m}(y)\}\) and
\(\ilc{m}\{x\gets \ilc{m}(y)\}\) are extensionally equal, and hence $p_2$ and
$p_3$ are eutt after interpretation.

\subsection{Nondeterminism}
\label{sec:nondeterminism}

While the story above is clean and satisfying for stateful effects,
nondeterminism is much more challenging.  Suppose we extend \imp with a
branching operator \cbr{p}{q} whose semantics is to nondeterministically pick a
branch to execute. This new feature is modeled very naturally using a
boolean-indexed \ebr event, creating a binary branch in the tree.  The new event
signature, a sample use, and the corresponding tree are shown below:

\noindent
\begin{minipage}[t]{2in}
\begin{lstlisting}[style=customcoq,aboveskip=-1.4\medskipamount]
  Variant Flip : Type -> Type :=
  | flip : Flip bool.
\end{lstlisting}
\end{minipage}
\hfill
\begin{minipage}[t]{2.5in}
  \ilc{Vis flip (fun b => if b then p else q)}
\end{minipage}
\hfill
\begin{minipage}[c]{.65in}
\small
\begin{forest}
    [\ebr
    [$\sembra{p}$]
    [$\sembra{q}$]]
\end{forest}
\end{minipage}

Naturally, as with memory events, \ebr does not come with its expected algebra:
associativity, commutativity and idempotence.
To recover these necessary equations to establish program equivalences
such as $p_3 \equiv \cbr{p_2}{p_3},$ we need to find a suitable
monad to interpret \ebr into.

\citeauthor{zakowski2021} used this approach in the Vellvm project
\cite{zakowski2021} for formalizing the nondeterministic features of the \ir.
Their model consists of a propositionally-specified set of computations:
ignoring other effects, the monad they use is \ilc{itree E _ -> Prop}.
The equivalence they build on top of it essentially amounts to a form of
bijection up-to equivalence of the contained monadic computations.
However, this approach suffers from several drawbacks.
First, one of the monadic laws is broken: the \ilc{bind} operation does not associate to the left.
Although stressed in the context of Vellvm~\cite{zakowski2021} and Yoon et al.'s work on layered
monadic interpreters~\cite{YoonZZ22},
this issue is not specific to \itreesn but rather to an hypothetical ``Prop Monad Transformer'',
i.e. to "fun M X => M X -> Prop", as pointed out previously in~\cite{MHRV20}.
The definition is furthermore particularly difficult to work with. Indeed, the corresponding monadic
equivalence is a form of bijection up-to setoid: for any trace in the source, we must existentially exhibit
a suitable trace in the target. The inductive nature of this existential is problematic: one usually cannot exhibit
upfront a coinductive object as witness, they should be produced coinductively.
Second, the approach is very much akin to identifying a communicating system with
its set of traces, except using a richer structure, namely monadic computations,
instead of traces: it forgets all information about \emph{when} nondeterministic choices are made.
As has been well identified by the process calculi tradition, and while trace equivalence is sometimes the
desired relation, such a model leads
to equivalences of programs that are too coarse to be compositional in general.
In a general purpose semantics library, we believe we should strive to provide as much compositional reasoning as possible and thus our tools  support both trace equivalence, and bisimulations~\citep{bloom1988}.
Third, because the set of computations is captured propositionally,
this interpretation is incompatible with the generation of an \emph{executable} interpreter by extraction, losing
one of the major strengths of  the \itreen framework.
\citeauthor{zakowski2021} work around this difficulty by providing two
interpretations of their nondeterministic events, and formally relating them.
But this comes at a cost --- the promise of a sound interpreter for free is broken --- and with constraints ---
non-determinism must come last in the stack of interpretations, and combinators whose equational theory is sensible
to nondeterminism are essentially impossible to define.

This difficulty with properly tackling nondeterminism extends also to
concurrency. \citeauthor{LesaniXKBCPZ22}  used \itreesn to prove
the linearizability of concurrent objects~\cite{LesaniXKBCPZ22}. Here too, they rely on sets of
linearized traces and consider their interleavings. While a reasonable
solution in their context, that approach strays from the monadic interpreter style
and fails to capture bisimilarity.

This paper introduces \ctreesn, a suitable monad for modeling nondeterministic
effects. As with \itreesn, the structure is compatible with divergence, external
interaction through uninterpreted events, extraction to executable reference
interpreters, and monadic interpretation.  The core intuition is based on the
observation that the tree-like structure from \itreesn is indeed the right one
for modeling nondeterminism. The problem arises from how the \itreesn definition
of \euttn observes which branch is taken, requiring that \textit{all} branches
be externally visible: while appropriate to model nondeterminism that results
from a lack of information, it does not correspond to true \textit{internal}
choice. Put another way, thinking of the trees as labeled transition systems, \itreesn
are \emph{deterministic}.  With \ctreesn, we therefore additionally consider
truly branching nodes, explicitly build the associated nondeterministic \lts,
and define proper bisimulations on the structure.

The resulting definitions are very expressive. As foreseen, they form a proper monad,
validating all monadic laws up-to coinductive structure equality, they allow us to
establish desired \imp equations such as $p_3 \equiv \cbr{p_2}{p_3}$, but
they also scale to model \ccs and cooperative multithreading.

Before getting to that, and to better motivate our definitions, let us further extend
our toy language with a \cstuck construction that cannot
reduce, and a \cprint instruction that simply prints a dot. We will refer to
this language as \impbr.
\begin{mathpar}
  \ccomm \defeq \cskip \sep \cassign{x}{e} \sep \cseq{c1}{c2} \sep \cwhile{b}{c}
  \sep \cbr{c1}{c2} \sep \cstuck \sep \cprint
\end{mathpar}
Consider the program $p \defeq \cbr{(\cwhile{\true}{\cprint})}{\cstuck}$.
Depending on the intended operational semantics associated with \cbrn, this program can
have one of two behaviors: (1) either to always reactively print an infinite chain
of dots, or (2) to become nondeterministically either similarly reactive, or completely
unresponsive.

When working with (small-step) operational semantics, the distinction between
these behaviors is immediately apparent in the reduction rule for \cbrn (we only show rules for the left branch here).
\begin{mathpar}
  \inferrule*[right=BrInternal]{\;}{\cbr{c1}{c2} \rightarrow c1}
  \and
  \inferrule*[right=BrDelayed]{c_1 \rightarrow c_1'}{\cbr{c_1}{c_2} \rightarrow c_1'}
\end{mathpar}

\textsc{BrInternal} specifies that \cbrn may simply reduce to the left branch,
while \textsc{BrDelayed} specifies that \cbrn can reduce to any state
reachable from the left branch.
From an observational perspective, the former situation describes a system where,
although we do not observe which branch has been taken, we do observe that
\emph{a} branch has been taken. On the contrary, the latter only progresses if one of the branches can progress, we thus directly observe the subsequent evolution of the chosen branch, but not the branching itself.

In order to design the right monadic structure allowing for enough flexibility
to model either behavior, it is useful to look ahead and anticipate how we will
reason about program equivalence, as described in detail in
Section~\ref{sec:bisim}.
The intuition we follow is to  interpret our computations as labeled transition systems and
define bisimulations over those, as is generally done in the process algebra literature.
From this perspective, the \imp program $p$ may correspond to three distinct LTSs
depending on the intended semantics, as shown in Figure~\ref{fig:choices}.
\begin{figure}
  \centering
\tikzset{
  ->, 
  >=stealth, 
  node distance=2cm, 
  every state/.style={
    minimum size=30pt,
    thick, fill=gray!10
  }, 
  every loop/.style={
    looseness=4, in=70,out=110,
  },
  initial text=$ $, 
}

\begin{subfigure}{0.32\textwidth}
  \centering
\begin{tikzpicture}
  \node[state] (1) {$p$};
  \node[state, right of=1] (2) {$reac$};
  \node[state, above of=1] (3) {$stuck$};

  \draw (1) edge[align=center] node[sloped]{\ebr\\\true}  (2);
  \draw (1) edge[align=center] node[sloped]{\ebr\\\false} (3);
  \draw (2) edge[loop above,align=center] node{\cprint\\\unit} (2);
\end{tikzpicture}
\caption{Observation: branches}
\label{subfig:vis}
\end{subfigure}
\hfill
\begin{subfigure}{0.32\textwidth}
  \centering
\begin{tikzpicture}
  \node[state] (1) {$p$};
  \node[state, right of=1] (2) {$reac$};
  \node[state, above of=1] (3) {$stuck$};

  \draw (1) edge[below] node[sloped]{$\tau$}  (2);
  \draw (1) edge[above] node[sloped]{$\tau$} (3);
  \draw (2) edge[loop above, align=center] node{\cprint\\\unit} (2);
\end{tikzpicture}
\caption{Observation: branching}
\label{subfig:tau}
\end{subfigure}
\hfill
\begin{subfigure}{0.32\textwidth}
  \centering
\begin{tikzpicture}
  \node[state] (1) {$p$};
  \node[state, right of=1] (2) {$reac$};

  \draw (1) edge[align=center] node{\cprint\\\unit}  (2);
  \draw (2) edge[loop above,align=center] node{\cprint\\\unit} (2);
\end{tikzpicture}
\caption{Observation: none}
\label{subfig:none}
\end{subfigure}
\caption{Three possible semantics for the program $p$, from an LTS perspective}
\label{fig:choices}
\end{figure}
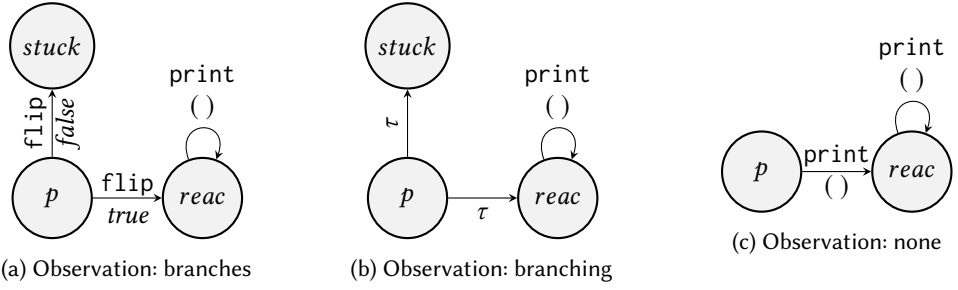

Figure~\ref{subfig:vis} describes the case where picking a branch is an unknown external event, hence where taking a specific branch is an \emph{observable} action with a dedicated label:
this situation is naturally modeled by a \visn node in the style of
\itreesn, that is
$\sembra{\cbr{p}{q}} \defeq \Vis{\ebr}{\lambda b \cdot \text{if } b \text{ then }
  \sembra{p} \text{ else } \sembra{q}}.$

Figure~\ref{subfig:tau} corresponds to \textsc{BrInternal}: both the stuck and
the reactive states are reachable, but we do not
observe the label of the  transition. This transition exactly corresponds to the internal
$\tau$ step of process algebra.
This situation is captured by introducing a new kind of node in our data structure,
a \brSn branch,
that maps in our bisimulations defined in Section~\ref{sec:bisim} to
a nondeterministic \textit{internal step}.
For this semantics, we thus have $\sembra{\cbr{p}{q}} \defeq \brS{2}{\lambda b \cdot \text{if } b \text{ then }
  \sembra{p} \text{ else } \sembra{q}}.$\footnote{The 2 indicates the arity of
  the branching.}

Figure~\ref{subfig:none} corresponds to
\textsc{BrDelayed} but raises the question: how do we build such a behavior?
A natural answer can be to
assume that it is the responsibility of the model, i.e., the function mapping
\imp's syntax to the semantic domain, \ctreesn, to explicitly build this LTS.
Here, $\sembra{p}$ would be an infinite sequence of \visn
\cprint nodes, containing no other node.  While that would be
convenient for developing the meta-theory of \ctreesn, this design choice would
render them far less compositional (and hence less useful) than we want them to
be. Indeed, we want our models to be defined as computable functions by
recursion on the syntax, whenever possible. But to build this LTS directly, the
model for \cbr{p}{q} needs to introspect the models for \sembra{p} and
\sembra{q} to decide whether they can take a step, and hence whether it
should introduce a branching node. But, in general, statically determining
whether the next reachable instruction is \cstuck is intractable, so that
introspection will be hard (or impossible) to implement. We thus extend \ctreesn
with a third category of  node, a \brDn node, which does not directly
correspond to states of an LTS.  Instead, \brDn nodes aggregate sub-trees
such that the (inductively reachable) \brDn children of a \brDn node are ``merged''
in the LTS view of the \ctreen.
This design choice means that, for the \textsc{BrDelayed} semantics, the model
is again trivial to define: $\sembra{\cbr{p}{q}} \defeq \brD{2}{\lambda b \cdot \text{if } b \text{ then }
  \sembra{p} \text{ else } \sembra{q}}$, but the definition of
bisimilarity for \ctreesn ensures that the behavior of $\sembra{p}$ is exactly
the LTS in Figure~\ref{subfig:none}.

Although \textsc{BrDelayed} and its corresponding LTS in
Figure~\ref{subfig:none} is \textit{one} example in which \brDn nodes are
needed, we will see another concrete use in modeling \ccs in
Section~\ref{sec:ccs}.  Similar situations arise frequently, for modeling
mutexes, locks, and other synchronization mechanisms (e.g., the ``await''
construct in Encore~\cite{encore}), dealing with crash failures in
distributed systems, encoding relaxed-consistency shared memory models (e.g.,
the \textsc{ThreadPromise} rule of promising
semantics~\cite{promising}).  The \brDn construct is needed whenever the
operational semantics includes a rule whose possible transitions depend on the
existence of other transitions, i.e. for any rule of the shape shown below, when
it may sometimes be the case (usually due to nondeterminism) that
$P \not \rightarrow$:
\[
  \inferrule*[right=ContingentStep]{P \rightarrow P'}{C[P] \rightarrow C[P']}
\]
The point is that to implement the LTS corresponding to \textsc{ContingentStep}
without using a \brDn node would require introspection of $P$ to determine
whether it may step, which is potentially non-computable.  The \brDn node
bypasses the need for that introspection at \textit{representation time},
instead pushing it to the characterization of the \ctreen as an LTS, which is
used only for \textit{reasoning} about the semantics.  The presence of the \brDn
nodes retains the constructive aspects of the model, in particular the ability
to \textit{interpret} these nodes at later stages, for instance to obtain an
executable version of the semantics.


%% file: core.tex

\subsection{Core definitions}

\begin{figure}
  \begin{lstlisting}[style=customcoq,basicstyle=\small\ttfamily]
  (* Core datatype *)
  CoInductive ctree (E : Type -> Type) (R : Type) :=
  | Ret (r : R)                                         (* a pure computation *)
  | Vis {X : Type} (e : E X) (k : X -> ctree)           (* an external event *)
  | brS (n : nat)  (k : fin n -> ctree)                 (* stepping branching *)
  | brD (n : nat)  (k : fin n -> ctree)                 (* delayed branching *)

  (* Bind, sequencing computations *)
  CoFixpoint bind {E T U} (t : ctree E T) (k : T -> ctree E U) : ctree E U :=
  match u with
  | Ret r    => k r
  | Vis e h => Vis e (fun x => bind (h x) k)
  | brS n h => brS n (fun x => bind (h x) k)
  | brD n h => brD n (fun x => bind (h x) k)
  end

  (* Unary guards *)
  Definition Guard (t : ctree E R) : ctree E R := brD 1 (fun _ => t)
  Definition Step    (t : ctree E R) : ctree E R := brS 1 (fun _ => t)

  (* Main fixpoint combinator *)
  CoFixpoint iter {I: Type} (body : I -> ctree E (I + R)) : I -> ctree E R :=
    bind (body i) (fun lr => match lr with
                                  | inr r => Ret r
                                  | inl i => Guard (iter body i)
                                  end)

  Notation "E ~> F" := (forall X, E X -> F X)
  (* Atomic ctrees triggering a single event *)
  Definition trigger : E ~> ctree E := fun R (e : E R) => Vis e (fun x => Ret x)
  (* Atomic branching ctrees *)
  Definition brS : ctree E (fin n) := fun n => brS n (fun x => Ret x)
  Definition brD : ctree E (fin n) := fun n => brD n (fun x => Ret x)
  \end{lstlisting}

  \caption{\ctreesn: definition and core combinators (\linkt{Core/CTreeDefinitions.v}{42})}
  \label{fig:defs}

\end{figure}

We are now ready to define our core datatype, displayed in the upper part of
Figure~\ref{fig:defs}.
The definition remains close to a coinductive implementation of the free monad, but hardcodes
support for an additional effect: unobservable, nondeterministic branching.
The \ctreen datatype, much like an \itreen, is parameterized by a signature of
(external) events \ilc{E} encoded as a family of types, and a return type \ilc{R}.
It is defined as a coinductive tree\footnote{The actual implementation uses
  a negative style with primitive projections. We omit this technical detail in
  the presentation.}
with four kind of nodes: pure computations (\ilc{Ret}),
external events (\ilc{Vis}),
internal branching with implicitly associated $\tau$ step (\ilc{brS}),
and delayed internal branching (\ilc{brD}).
The continuation following external events is indexed by the return type
specified by the emitted event.
For the sake of simpler exposition, we restrict both internal
branches to be of finite width: their continuation are
indexed by finite types \ilc{fin}, but this is not a fundamental
limitation.\footnote{We actually also support another branch
  that defines \emph{enhanced \ctreesn}, with arbitrary branching
  over indexes specified in the type of the data structure by an
  interface akin to \ilc{E}. It is not linked here for anonymity reasons, but
  should the paper be accepted, we will add it to the artifact. It is only
  used for the material described in this paper for Lemma~\ref{lem:refine}.}
When  using finite branching, we abuse notation and write, for instance,
$\brS{2}{t~u}$ for the computation branching with two branches, rather
than explicitly spelling out the continuation which  branches on the choice
index: $\ilc{fun i => match i with 0 => t | 1 => u end}$.


The remainder of Figure~\ref{fig:defs} displays (superficially simplified)
definitions of the core combinators.
As expected, \ilc{ctree E} forms a monad for any interface \ilc{E}: the
\ilc{bind} combinator simply lazily crawls the potentially infinite first tree,
and passes the value stored in any reachable leaf to the continuation.
The \ilc{iter} combinator is central to encoding looping and recursive features:
it takes as argument a body, \ilc{body}, intended to be iterated, and is defined
such that the computation returns either a new index over which to continue iterating, or
a final value; \ilc{iter} ties the recursive knot.
Its definition is analogous to the one for \itreesn, except that we need to ask
ourselves how to guard the \ilc{cofix}: if \ilc{body} is a constant, pure
computation, unguarded corecursion would be ill-defined. \itreesn use a \delay node for this purpose.
Here, we instead use a unary delayed branch
as a guard, written \guard. We additionally
write \step for the unary stepping branch---we will discuss how they relate in Section~\ref{subsec:itree}.
The minimal computations respectively triggering an event \ilc{e}, generating observable
branching, or delaying a branch, are defined as \ilc{trigger}, $Br_S^n$, and $Br_D^n$.

Convenience in building models comes at a cost: many \ctreesn represent the same LTS.
Figure~\ref{fig:stuck-spin} illustrates this  by defining
several \ctreesn implementing the stuck LTS and the silently spinning one.
The \ilc{stuckS} and \ilc{stuckD} correspond to internal choices with no outcome, while \ilc{stuckE}
is a question to the environment that cannot be answered, leaving the computation hanging;
the \ilc{spinD*} trees are infinitely deep, but never find in their structure a transition to take.
We define formally the necessary equivalences on computations to prove this
informal statement in Section~\ref{sec:bisim}.
The astute reader may wonder whether both kind of branching nodes
really are necessary. While convenient, we show
in Section~\ref{subsec:guarded} that $\brSn$
can  be expressed in terms of $\brDn$ and $\step$.

\begin{figure}
  \tikzset{
    ->, 
    >=stealth, 
    node distance=2cm, 
    every state/.style={
      minimum size=30pt,
      thick, fill=gray!10
    }, 
    every loop/.style={
      looseness=4, in=70,out=110,
    },
    initial text=$ $, 
}
\begin{minipage}{0.75\textwidth}
  \begin{lstlisting}[style=customcoq,basicstyle=\small\ttfamily]
  (* Stuck processes *)
  Definition stuckE (e : E void) : ctree E void := trigger e
  Definition stuckS : ctree E void := brS 0
  Definition stuckD : ctree E void := brD 0
  CoFixpoint spinD  : ctree E R := Guard spinD
  CoFixpoint spinD_nary n : ctree E R := brD n ;; spinD_nary n
\end{lstlisting}
\end{minipage}
\begin{minipage}{0.2\textwidth}
  \hfil\begin{tikzpicture}
    \node[state] (1) {$stuck$};
  \end{tikzpicture}
\end{minipage}

\begin{minipage}{0.75\textwidth}
  \begin{lstlisting}[style=customcoq,basicstyle=\small\ttfamily]
  (* Spinning processes *)
  CoFixpoint spinS : ctree E R := Step spinS
  CoFixpoint spinS_nary n : ctree E R := brS n ;; spinS_nary n
\end{lstlisting}
\end{minipage}
\begin{minipage}{0.2\textwidth}
  \hfil\begin{tikzpicture}
    \node[state] (1) {$spin$};
    \draw (1) edge[loop above,align=center] node[sloped]{$\tau$} (1);
  \end{tikzpicture}
\end{minipage}

\caption{Concrete representations of stuck and spinning LTSs}
\label{fig:stuck-spin}
\end{figure}
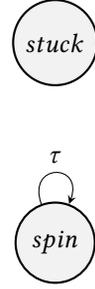


\subsection{A hint of introspection: heads of computations}
\label{subsec:head}

\brDn nodes avoid the need for introspection on trees to model something as generic as a delayed
branching construct such as the one specified by \textsc{BrDelayed}.
However, introspection becomes necessary to build a tree that depends on the
reachable external actions of the sub-trees.
This is the case for example for \ccs's parallel operator that we model in Section~\ref{sec:ccs}.
The set of reachable external actions is not computable in general, as we may have to first know if the
computation to the left of a sequence terminates before knowing if the events
contained in the continuation are reachable.
We are, however, in luck, as we have at hand a semantic domain able to represent
potentially divergent computations: \ctreesn themselves!

\begin{figure}
  \begin{lstlisting}[style=customcoq,basicstyle=\small\ttfamily]
  Variant action E R :=
    | ARet (r : R)
    | ABr  {n} (k : fin n -> ctree E R)
    | AVis {X} (e : E X) (k : X -> ctree E R).

  CoFixpoint head {E X} (t : ctree E X) : ctree E (action E X) :=
    match t with
    | Ret x => Ret (ARet x) | Vis e k => Ret (AVis e k) | brS n k => Ret (ABr k)
    | BrD n k => BrD n (fun i => head (k i))
    end.
  \end{lstlisting}

  \caption{Lazily computing the set of reachable observable nodes (\linkt{Misc/Head.v}{1})}
\label{fig:head}
\end{figure}

The \ilc{head} combinator, described on Figure~\ref{fig:head}, builds a
pure, potentially diverging computation
\emph{only made of delayed choices}, and whose leaves contain all reachable
subtrees starting with an observable node.
These ``immediately'' observable trees are captured in an \ilc{action} datatype, which is used
as the return type of the built computation. The \ilc{head}\footnote{The \ilc{head t} computation could be typed at the empty interface. It is in practice simpler to directly type it at the same interface as the \ilc{t}.} combinator simply crawls the
tree by reconstructing all delayed branches, until it reaches a subtree with any other node at its root; it then returns that subtree as the corresponding \ilc{action}.


%% file: bisim.tex
Section~\ref{sec:ctrees} introduced \ctreesn, the domain of computations we consider,
as well as a selection of combinators upon it. We now turn to the question of
comparing computations represented as \ctreesn for notions of equivalence
and refinement. In particular, we formalize the LTS
representation of a \ctreen that we have followed to justify our definitions.
This section first introduces a (coinductive) syntactic equality of
\ctreesn, then it recovers the traditional notions of strong and weak
bisimilarity of processes, as well as trace equivalence, for an LTS
 derived from the tree structure.  We equip these notions with a
primitive equational theory for \ctreesn; this theory provides
the building blocks necessary for deriving domain-specific equational
theories, such as the ones established in Section~\ref{sec:ccs} for \ccs, and in
Section~\ref{sec:yield} for cooperative scheduling.

\subsection{Coinductive proofs and up-to principles in Coq}

Working with \ctreesn requires the use of several coinductive predicates and
relations to describe equivalences, refinements and invariants.
Doing so at scale in \coq would be highly impractical if only using its native
support for coinductive proofs, but it is nowadays more feasible thanks to
library support~\cite{paco,gpaco,companion}.
Our development relies on Pous's \texttt{coinduction}
library~\cite{coq-coinduction} to
define and reason about the various coinductive predicates and relations we
manipulate. We briefly recall the essential facilities, based
on the \emph{companion}~\cite{companion}, that the library provides.
In particular, we define a subset of candidate ``up-to'' functions---we indicate
in the subsequent subsections which ones constitute valid up-to principles for
respectively the structural equality and the strong bisimulation we manipulate.
Note: this section is aimed at the reader interested in understanding the internals of
our library: it can be safely skipped at first read.

The core construction provided by the library is a \gfpn operator (\gfptop{b}\ilc{:X})
for any complete lattice \ilc{X}, and monotone endofunction $b : \ilc{X} \to \ilc{X}$.
In particular, the sort of \coq propositions \ilc{Prop} forms a
complete lattice, as do any function
from an arbitrary type into a complete lattice---coinductive relations, of
arbitrary arity, over arbitrary types, can therefore be built using this combinator.
We most often instantiate \ilc{X} with the complete
lattice of binary relations over \ctreesn: \ilc{X := ctree E A -> ctree E A ->
  Prop} for fixed parameters \ilc{E} and \ilc{A}.
We write such binary relations as \ilc{rel(A,B)} for
\ilc{A -> B -> Prop}, and \ilc{rel(A)} for \ilc{rel(A,A)};
for instance: \ilc{X := rel(ctree E A)}.

The library provides tactic support for   coinductive proofs based on
Knaster-Tarski's theorem: any \pfpn is below the \gfpn.
Specialized to relations on \ctreesn, the proof
method consists in exhibiting a relation \ilc{R}, that can be thought of as a
set of pairs of trees, providing a ``coinduction candidate'', and
proving that \ilc{(forall t u, R t u -> b R t u) -> forall t u, R t u -> gfp b t
  u}.
The major benefit of the companion is, however, to further provide support for
proofs by \emph{enhanced coinduction}.

Given an endofunction $b$, a (sound) enhanced coinduction principle,
also known as an \emph{up-to principle}, relies on an additional function
$f$\ilc{: X -> X} allowing one to work with \ilc{bf} (the composition of $b$ with $f$) instead of $b$:
any \pfpn of $bf$ is below the \gfpn of $b$. Concretely, the user has now access to a new proof principle.
Rather than having to ``fall back'' exactly into their coinduction hypothesis after ``stepping''
through $b$, they may first apply $f$. In the case of a coinductive relation,
simple examples of up-to principles include adding the diagonal or swapping the arguments,
allowing one to conclude by reflexivity or invoke symmetry regardless of the coinduction candidate considered.

Significant effort has been invested in identifying classes of sound up-to
principles, and developing ways to combine them~\cite{sangiorgi98,pous07,SR12}.
The companion relies on one particular sub-class of sound principles that forms a complete lattice,
the \emph{compatible} functions: we hence can conduct
proofs systematically up-to the greatest compatible function, dubbed the
\textit{companion} and written \compa{b}.
In practice, this means that
 a coinductive proof can, on-the-fly, use any
valid up-to principle drawn from the companion.
We refer the interested reader to the literature~\cite{companion} and
our formal development for further details, and
describe below the  specialization to \ctreesn of the up-to functions we use.

\begin{figure}
\begin{mathpar}
  \reflC R \defeq \{(x,x)\} \and
  \symC R \defeq \{(y,x)~\mid~R~x~y\} \and
  \transC R \defeq \{(x,z)~\mid~\exists y,~R~x~y\land R~y~z\} \and
  \visC R \defeq \{(\Vis{e}{k},\Vis{e}{k'})~\mid~\forall v,~R~(k~v)~(k'~v)\}\and
  \bindC{\sc{equiv}}~ R \defeq \{(x\bindc k,y\bindc l)~\mid~\sc{equiv}~x~y \land \forall v,R~(k~v)~(l~v)\}\and
  \uptorelC{\sc{equiv}}~ R \defeq \{(x,y)~\mid~\exists
  x'~y',~\sc{equiv}~x~x'\land R~x'~y'\land~\sc{equiv}~y'~y\}
\end{mathpar}
\caption{Main generic up-to principles used for relations of \ctreesn where
  \ilc{R : rel (ctree E X)}}
\label{fig:upto}
\end{figure}

Figure~\ref{fig:upto} describes the main generic up-to principles we use for our
relations on \ctreesn.\footnote{These are the core examples of library level up-to principles we provide. For our \ccs case study, we also prove the traditional language level ones.}
Note that since we are considering relations on \ctreesn, each of these principles, these candidates ``$f$'', are endofunctions of relations: for instance, \reflC is the constant diagonal relation, \symC builds the symmetric relation, \transC is the composition of relations.
The validity of \reflC, \symC and \transC
for a given endofunction $b$ entails respectively the reflexivity, symmetry and
transitivity of the relations $(b\compa{b}\;R)$ and $(\compa{b}\;R)$.
These two relations are precisely the ones involved during a proof by
coinduction up-to companion: the former as our goal, the latter as our
coinduction hypothesis.
The \visC and \bindC{\_} up-to functions help when reasoning structurally, respectively
allowing to cross through \visn nodes and \texttt{bind} constructs during proofs by coinduction.
Finally, validity of the \uptorelC{\sc{equiv}} principle allows for rewriting
via the $equiv$ relation during coinductive proofs for $b$.

\subsection{Coinductive equality for \ctreen}

Coq's equality, \ilc{eq}, is not a good fit to express the structural equality
of coinductive structures---even the eta-law for a coinductive data structure does
not hold up-to \ilc{eq}. We therefore define, as is standard, a structural
equality\footnote{Note that for \itreesn, this relation corresponds to what \citeauthor{XZHH+20}
  dub as \emph{strong bisimulation}, and name \ilc{eq_itree}. We carefully avoid
this nomenclature here to reserve this term for the relation we define in Section~\ref{subsec:bisim}}
by coinduction \ilc{equ: rel(ctree E A)} (written \equ{} in infix).
The endofunction simply matches head constructors and behaves extensionally on continuations.

\begin{definition}{Structural equality (\linkt{Eq/Equ.v}{32})}
\begin{align*}
  \equn \defeq \gfp{R}
  & \{(\Ret{v},~\Ret{v})\}~\cup
    \{(\Vis{e}{k},~\Vis{e}{k'})~\mid~\forall v,~\gfpc{R}~(k~v)~(k'~v)\} ~\cup\\
  &\{(\brD{n}{k},~\brD{n}{k'})~\mid~\forall v,~\gfpc{R}~(k~v)~(k'~v)\} \cup
    \{(\brS{n}{k},~\brS{n}{k'})~\mid~\forall v,~\gfpc{R}~(k~v)~(k'~v)\}
\end{align*}
\end{definition}

The \equn relation raises no surprises: it is an equivalence relation, and is
adequate to
prove all eta-laws---for the \ctreen structure itself and for the \cofixes
we manipulate. Similarly, the usual monadic laws are established with respect to \equn.

\begin{lemma}{Monadic laws (\linkt{Eq/Equ.v}{838})}
\begin{mathpar}
  \Ret v \bindc k \equ k~v \and x \gets t\bind \Ret x \equ t \and (t \bindc k) \bindc l \equ t \bindc (\lambda x \Rightarrow k~ x \bindc l)
\end{mathpar}
\end{lemma}


Of course, formal equational reasoning with respect to an equivalence relation
other than \ilc{eq} comes at the usual cost: all constructions introduced over
\ctreesn must be proved to respect \equn (in \coq parlance, they must be
\ilc{Proper}), allowing us to work  painlessly with setoid-based
rewriting.

Finally, we establish some enhanced coinduction principles for \equn.
\begin{lemma}{Enhanced coinduction for \equn (\linkt{Eq/Equ.v}{167})}

$\reflC$, $\symC$, $\transC$, $\bindC{\equ}$ and
$\uptorelC{\equ}$ provide valid up-to principles for $\equn$.
\end{lemma}

While $\equn$, being a structural equivalence, is very comfortable to work with,
it, naturally, is much too stringent.  To reason semantically about \ctreesn, we need
a relation that remains termination sensitive but allows for differences in
internal steps, that still imposes a tight correspondence over external events,
but relaxes its requirement for nondeterministically branching nodes.  We achieve this by drawing
from standard approaches developed for process calculi.

\subsection{Looking at \ctreesn under the lens of labeled transition systems}
\label{subsec:LTS}

To build a notion of bisimilarity between \ctreen computations, we
associate a labeled transition system to a \ctreen, as defined in
Figure~\ref{fig:lts}.
This LTS exhibits three kinds of labels: a \ltau\footnote{We warn again the reader accustomed to \itreesn to think of \ltau under the lens of the process algebra literature, and not as a representation of \itreen's \ilc{Tau} constructor.}
witnesses a stepping branch,
an \lobs{e}{x} observes the encountered event together with the answer from the
environment considered, and a \lret{v} is emitted when returning a value.
Interestingly, there is a significant mismatch between the structure of the
tree and the induced LTS: the states of the LTS correspond to the nodes of the \ctreen \emph{that are not
  immediately preceded by a delayed choice}.
Accordingly, the definition of the transition relation between states inductively iterates over delayed branches.
Stepping branches and visible nodes map immediately to a set of transitions, one for each outgoing edge; finally
a return node generates a single \ret{} transition, moving onto a stuck state,
encoded as a nullary branching node and written \stuck.
These rules formalize the intuition we gave in Section~\ref{sec:nondeterminism} that allowed us to derive the LTSs of Figure~\ref{fig:choices}  from the corresponding \impbr terms.


\begin{figure}
  \begin{mathpar}
    \llabel \defeq \ltau \sep \lobs{e}{v} \sep \lret{v}\\
    \inferrule{\lstep{k~v}{l}{t}}{\lstep{\brD{n}{k}}{l}{t}}\and
    \inferrule{~}{\lstep{\brS{n}{k}}{\ltau}{k~v}}\and
    \inferrule{~}{\lstep{\Vis{e}{k}}{\lobs{e}{v}}{k~v}}\and
    \inferrule{~}{\lstep{\Ret{v}}{\lret{v}}{\stuck}}
  \end{mathpar}
  \caption{Inductive characterisation of the LTS induced by a \ctreen (\linkt{Eq/Trans.v}{69})}
  \label{fig:lts}
\end{figure}

Defining the property of a tree to be \textit{stuck}, that is: $\isstuck{t} \ \defeq \ \forall l~u,~\lnot (\lstep{t}{l}{u})$, we can make
the depictions from Figure~\ref{fig:stuck-spin} precise: nullary
nodes are stuck by construction, since stepping would require a
branch, while \spinDn is proven to be stuck by induction, since it cannot reach any
step. %
  \begin{mathpar}
    \isstuck{\brS{0}}
    \and
    \isstuck{\brD{0}}
    \and
    \isstuck{\spinDn~n}
  \end{mathpar}

The stepping relation interacts slightly awkwardly with \bindn: indeed,
although a unit for \bindn, the \ilc{Ret} construct is not inert from the
perspective of the LTS. Non \lretKW{} transitions can therefore be propagated
below the left-hand-side of a \bindn, while a \lretKW{} transition in the prefix
does not entail the existence of a transition in the
\bindn---Figure~\ref{lemma-bind} describes the corresponding lemmas.

\begin{figure}
  \begin{mathpar}
    \inferrule{\lstep{t}{l}{u} \quad l \not = \lret{v}}{\lstep{t \bindc k}{l}{u \bindc k}}
    \and
    \inferrule{\lstep{t}{\lret{v}}{\stuck} \quad \lstep{k~v}{l}{u}}{\lstep{t \bindc k}{l}{u}}
    \and
    \inferrule{\lstep{t \bindc k}{l}{u}}
    {(l \not = \lret{v} \land \exists
      t',\lstep{t}{l}{t'} \land u \equ t' \bindc k) \lor (\exists v,
      \lstep{t}{\lret{v}}{\stuck}\land \lstep{k~v}{l}{u})}
  \end{mathpar}
  \caption{Transitions under bind (\linkt{Eq/Trans.v}{953})}
  \label{lemma-bind}
\end{figure}


We additionally define the traditional \emph{weak transition} $\wstep{s}{l}{t}$ on
the LTS that can perform tau transitions before or after the $l$ transition
(and possibly be none if $l=\tau$).
This part of the theory is so standard that we can directly reuse parts of the
development for \ccs that Pous developed to illustrate the
companion~\cite{coq-coinduction-example},
with the exception that we need to work in a Kleene Algebra with a model closed
under \equn rather than \ilc{eq}.

\subsection{Bisimilarity}
\label{subsec:bisim}

Having settled on the data structure and its induced LTS, we are back on a
well-traveled road: strong bisimilarity (referred simply as
bisimilarity in the following) is defined in a completely standard way
over the LTS view of \ctreesn.

\begin{definition}[Bisimulation for \ctreesn (\linkt{Eq/SBisim.v}{62})]
  The progress function $sb$ for  bisimilarity maps a relation
  $\mathcal{R}$ over \ctreesn\ to the relation such that $sb~R~s~t$ holds if and only if:\\
  \begin{minipage}[c]{.6\textwidth}
    \[
      \forall l~s',~ \lstep{s}{l}{s'}
      \implies
      \exists t'.~ s' ~\mathcal{R}~ t'
      \land \lstep{t}{l}{t'}
    \]
    \begin{center}
      and conversely
    \end{center}
    \vspace{-1ex}
    \[
      \forall l~t',~ \lstep{t}{l}{t'}
      \implies
      \exists s'.~ s' ~\mathcal{R}~ t'
      \land \lstep{s}{l}{s'}
    \]
    \vspace{.5ex}
  \end{minipage}
  i.e. \qquad
  \begin{minipage}[c]{.35\textwidth}
    \begin{tikzpicture}
      \node (s) at (-1,1) {$s$};
      \node (R) at (0,1) {$sb~\mathcal{R}$};
      \node (t) at (1,1) {$t$};
      \node (sp) at (-1,0) {$s^\prime$}
      edge [<-] node[auto] {$l$} (s);
      \node (Rp) at (0,0) {$\mathcal{R}$};
      \node (tp) at (1,0) {$t^\prime$}
      edge [<-] node[auto,swap] {$l$} (t);
    \end{tikzpicture}
  \end{minipage}

  Bisimilarity, written $s \sbisim t$, is defined as the
  \gfpn of $sb$: $\sbisimn \defeq \gfptop{sb}.$
\end{definition}

All the traditional tools
surrounding bisimilarity can be transferred to our setup.
We omit the details for spacing concerns, but additionally provide:
\begin{itemize}
\item weak bisimilarity, written $s \wbisim t$, derived from the
  definition of the weak transition (\linkt{Eq/WBisim.v}{88});
\item a characterization of the traces (represented as colists) of a \ctreen,
  used to define trace-equivalence (written $\treq$) (\linkt{Eq/Trace.v}{15});
\item strong simulations (written $\ssim$), defined as the \gfpn of the half
  game for strong bisimulation (\linkt{Eq/SSim.v}{31}).
\end{itemize}


\subsubsection{Core equational theory}

\begin{figure}
  \begin{mathpar}
    {\mprset { fraction ={===}}
      \inferrule {x = y} {\Ret{x} \sbisim \Ret{y}}
    }
    \and
    {\mprset { fraction ={===}}
      \inferrule {\forall x,~h~x\sbisim k~x} {\Vis{e}{h} \sbisim \Vis{e}{k}}
    }
    \and
    {\mprset { fraction ={===}}
      \inferrule{(\forall x,\exists y,~h~x\sbisim k~y) \land(\forall y,\exists x,~h~x\sbisim k~y)}{\brS{n}{h}\sbisim\brS{m}{k}}
    }
    \and
    \inferrule{(\forall x,\exists y,~h~x\sbisim k~y) \land(\forall y,\exists x,~h~x\sbisim k~y)}{\brD{n}{h}\sbisimR\brD{m}{k}}
    \and
    \inferrule{t \sbisim u \and (\forall x,~g~x\sbisim k~x)}{t\bindc g \sbisim u \bindc k}
    \\
    \guard~t\sbisim t
    \and
    \inferrule{\isstuck{u}}{\brDtwo{t}{u} \sbisim t}
    \and
    \brDtwo{t}{(\brDtwo{u}{v})} \sbisim \brDtwo{(\brDtwo{t}{u})}{v}
    \and
    \brDtwo{t}{u} \sbisim \brDtwo{u}{t}
    \and
    \brDtwo{t}{t} \sbisim t
    \and
    \brDtwo{(\brDtwo{t}{u})}{v} \sbisim \brDthree{t}{u}{v}
    \and
    \brStwo{t}{u} \sbisim \brStwo{u}{t}
    \and
    \brStwo{t}{t} \sbisim \step~t
    \and
    \step~t\wbisim t
    \and
    \spinDn~n \sbisim \spinDn~m
    \and
    \inferrule{(n>0 \land m>0)\lor (n = m = 0)}{\spinSn~n \sbisim \spinSn~m}
    \and
  \end{mathpar}

  \caption{Elementary equational theory for \ctreesn (\linkt{Eq/SBisim.v}{1008})}
  \label{fig:sbisim-laws}
\end{figure}

Bisimilarity forms an equivalence relation satisfying a collection of
primitive laws for \ctreesn summed up in Figure~\ref{fig:sbisim-laws}.
We use simple inference rules to represent an implication from the
premises to the conclusion, and double-lined rules to represent equivalences.
Each rule is proved as a lemma with respect to the definitions above.

The first four rules recover some structural reasoning on the syntax of the
trees from its semantic interpretation. These rules are much closer to what
\euttn provides by construction for \itreesn:
leaves are bisimilar if they are equal, and computations performing the same
external interaction must remain point-wise bisimilar.
Stepping branches, potentially of distinct arity, can be matched one against
another if and only if both domains of indexes can be injected into the other to
reestablish bisimilarity. In contrast, this condition is sufficient but not
necessary for delayed branches, since the points of the continuation
structurally immediately accessible do not correspond to accessible states in
the LTS. Finally, bisimilarity is a congruence for \bindn.

Another illustration of this absence of equivalence for delayed branching nodes
as head constructor is that such a computation may be strongly bisimilar to a
computation with a different head constructor.
The simplest example is that \sbisimn{} can ignore (finite numbers of) \guard
nodes.
Another example is that stuck processes behave as a unit for delayed branching nodes.
We furthermore obtain the equational theory that we  expect for
nondeterministic effects. Delayed branching is associative, commutative, idempotent,
and can be merged into delayed branching nodes of larger
arity w.r.t. \sbisimn{}.\footnote{Stating these facts generically in the arity of branching is quite
  awkward, we hence state them here for binary branching, but adapting them at
  other arities is completely straightforward.}
In contrast, stepping branches are only commutative, and almost idempotent,
provided we introduce an additional \step.
This \step can  in turn be ignored by moving to weak bisimilarity, making stepping branches commutative and properly idempotent; however, it crucially remains \emph{not} associative, a standard fact in process algebra.\footnote{This choice would be referred as external in this community}

Finally, two  delayed spins are always
bisimilar (neither process can step) while two stepping spins
are bisimilar if and only if they are  both
nullary (neither one can step), or both non-nullary.

We omit the formal equations for sake of space here, but we additionally prove
that the \ilc{iter} combinator deserves its name: the Kleisli category of the
\ilc{ctree E} monad is iterative w.r.t. strong bisimulation (\linkt{Eq/IterFacts.v}{1}).
Concretely, we prove that the four equations described in \cite{XZHH+20},
Section~4, hold true.
The fact that they hold w.r.t. strong bisimulation is a direct consequence of the design choice taken in our definition of \ilc{iter}: recursion is guarded by a \guard.
We conjecture that one could provide an alternate iterator guarding recursion by
a \step, and recover the iterative laws w.r.t. weak bisimulation, but have not
proved it and leave it as future work.

Naturally, this equational theory gets trivially lifted at the language level for \impbr (\linkimp{ImpBr.v}{125}).
The acute reader may notice that in exchange for being able to work with strong
bisimulation, we have mapped the silently looping program to
\ilc{spinD}, hence identifying it with stuck processes.
For an
alternate model observing recursion, one would need to investigate the use of
the alternate iterator mentioned above and work with weak bisimilarity.

\subsubsection{Proof system for bisimulation proofs}

As is usual, the laws in Figure~\ref{fig:sbisim-laws}, enriched with
domain-specific equations, allow for deriving further equations purely
equationally.  But to ease the proof of these primitive laws, as well
as new nontrivial equations requiring explicit bisimulation proofs, we
provide proof rules that are valid during bisimulation proofs.

\begin{figure}
  \begin{mathpar}
    \inferrule{}{\Ret{v}\sbisimR\Ret{v}}
    \\\\
    \inferrule{\forall v,~R~(k~v)~(k'~v)}{\Vis{e}{k}\sbisimR\Vis{e}{k'}}
    \and
    \inferrule{(\forall x,\exists y,~R~(k~x)~(k'~y)) \land(\forall y,\exists x,~R~(k~x)~(k'~y))}{\brS{n}{k}\sbisimR\brS{m}{k'}}
    \and
    \inferrule{\forall v,~R~(k~v)~(k'~v)}{\brS{n}{k}\sbisimR\brS{n}{k'}}
    \and
    \inferrule{(\forall x,\exists y,~(k~x)\!\sbisimR\!(k'~y)) \!\land\!(\forall y,\exists x,~(k~x)\sbisimR(k'~y))}{\brD{n}{k}\!\sbisimR\!\brD{m}{k'}}
    \and
    \inferrule{t\sbisimR u}{\guard~t\!\sbisimR\!\guard~u}
    \and
    \inferrule{t~R~u}{\step~t\!\sbisimR\!\step~u}
  \end{mathpar}
  \caption{Proof rules for coinductive proofs of \sbisimn (\linkt{Eq/SBisim.v}{673})}
  \label{fig:sbisim-upto}
\end{figure}
Given a bisimulation candidate $R$,
we write $t \sbisimR u$ for $sb~(\compa{sb}~R):$ one needs to ``play the
bisimulation game'' by crossing $sb,$ and may rely on $R$ to conclude by
coinduction, while furthermore being able to exploit the companion $\compa{sb}$
to leverage sound up-to principles.
We depict the main rules we use in Figure~\ref{fig:sbisim-upto}. These proof
rules are notably convenient because they avoid an
exponential explosion in the number of cases in our proofs, which otherwise
would arise due to the systematic binary split entailed by the natural way to
play the bisimulation game.
These rules essentially match up counterpart \ctreen constructors at the level of bisimilarity,
but additionally make a distinction as to whether applying the rule soundly acts
as playing the game---i.e., the premises refer to $R$, allowing to
conclude using the coinduction hypothesis---or whether they do not---, i.e., the
premises still refer to $\sbisimR$.
The latter situation arises when using the proof rules that strip off
delayed branches from the structure of our trees: on either side, they do not
entail any step in the corresponding LTSs, but rather correspond to recursive calls to its
inductive constructor.

Furthermore, we provide a rich set of valid up-to principles:
\begin{lemma}{Enhanced coinduction for \sbisimn (\linkt{Eq/SBisim.v}{219})}
  \label{lem:sbisim-upto}
  The functions $\reflC$, $\symC$, $\transC$, $\visC$, $\bindC{\sbisim}$,
  $\uptorelC{\equ}$ and $\uptorelC{\sbisim}$ provide valid up-to principles for $\sbisimn$.
\end{lemma}
In particular, the equation $\guard~t\sbisim t$ can be used for rewriting during
bisimulation proofs, allowing for asymmetric stripping of guards.

\subsection{Eliminating nondeterministic stepping choices}
\label{subsec:guarded}

Out of programming convenience, \ctreesn offer both stepping and delayed
branching of arbitrary arity.
However, as hinted at in Section~\ref{sec:ctrees}, this is
superfluous: a stepping branch can always be simulated by delaying the same
choice, and guarding each branch by a unary \step. We make this intuition formal
by proving that the following combinator, precisely performing this
transformation, preserves the behavior of the computation with respect to \sbisimn.

\begin{definition}[Elimination of nondeterministic stepping choice (\linkt{Misc/brSElim.v}{13})]
~
\begin{lstlisting}[style=customcoq]
Definition BrSElim {E X} (t : ctree E X) : ctree E X :=
  iter (fun t => match t with
                 | Ret  r   => Ret (inr r)
                 | brD  k   => x <- brd;; Ret (inl (k x))
                 | brS  k   => x <- brd;; Step (Ret (inl (k x)))
                 | Vis  e k => x <- trigger e;; Ret (inl (k x))
                 end) t.
\end{lstlisting}
\end{definition}

\begin{lemma}{Stepping nondeterministic branches can be eliminated (\linkt{Misc/brSElim.v}{249})}
\quad  \(\forall t,~\ilc{BrSElim}~t\sbisim t\)
\end{lemma}

The proof of this statement is quite illustrative of bisimulation proofs over
\ctreesn.
We provide its sketch \arxiv{in the appendix}{in the extended version~\cite{arxiv}} as a means to illustrate the layers of facilities put in place in the formal development
to allow for the translation of paper intuitions, most notably the enhanced
coinduction principles that the companion allows us to build. It also
illustrates the difficulties that the structure raises when bridging the distance between the syntax (the tree) and its implicit semantics (the LTS).


%% file: interp.tex
The \itreen ecosystem fundamentally relies on the incremental interpretation of
effects, represented as external events, into their monadic implementations.
Through this section, we show how \ctreesn fit into this narrative both by
supporting the interpretation of their own external events, and by being
a suitable target monad for \itreesn, for the implementation of nondeterministic
branching.

\subsection{Interpretation}
\label{subsec:interp}

\itreesn support interpretation: provided a \emph{handler} \ilc{h:E ~> M}
implementing its signature of events \ilc{E} into a suitable monad \ilc{M},
the \ilc{(interp h):itree E ~> M} combinator provides an implementation of
any computation into \ilc{M}.
The only restriction imposed on the target monad \ilc{M} is that it must support
its own \ilc{iter} combinator, i.e., be iterative, so that the full
coinductive nature of the tree can be internalized in \ilc{M}.
For this implementation to be sensible and amenable to verification in
practice, one must, however, check an additional property: \ilc{interp h}
should form a monad morphism---in particular, it should map \euttn
\itreesn to equivalent monadic computations in \ilc{M}.

\begin{figure}
  \begin{lstlisting}[style=customcoq,basicstyle=\small\ttfamily]
    Definition interp (h : E ~> M) : ctree E ~> M := fun R =>
      iter (fun t => match t with
      | Ret r   => ret (inr r)
      | BrD n k => bind (mBrD n) (fun x => ret (inl (k x)))
      | BrS n k => bind (mBrS n) (fun x => ret (inl (k x)))
      | Vis e k => bind (h e)       (fun x => ret (inl (k x)))
      end).
  \end{lstlisting}
  \caption{Interpreter for \ctreesn (class constraints omitted) (\linkt{Interp/Interp.v}{32})}
  \label{fig:interp}
\end{figure}

Unsurprisingly, given their structure, \ctreesn enjoy their own \ilc{interp}
combinator. Its definition, provided in Figure~\ref{fig:interp}, is very close
to its \itreen counterpart.  The interpreter relies on the \ctreen version of
\ilc{iter}, chaining the implementations of the external events in the process.
The target monad must naturally still be iterative, but must also explain how it
internalizes branching nodes through the \ilc{mBrD} and \ilc{mBrS} operations.

Perhaps more surprisingly, the requirement that \ilc{interp h} define a monad
morphism unearths interesting subtleties.
Let us consider the elementary case where the interface \ilc{E} is implemented in terms
of (possibly pure) uninterpreted computations, that is when \ilc{M := ctree F}.
The requirement becomes: $\forall t~ u, t \sbisim u \rightarrow \ilc{interp h}~t
\sbisim \ilc{interp h}~u.$
But this result does not hold for an arbitrary \ilc{h}: intuitively, our
definition for \sbisimn has implicitly assumed that implementations of external
events may eliminate reachable states in the computation's induced LTS---through pure implementations---but should not be allowed to introduce new ones.

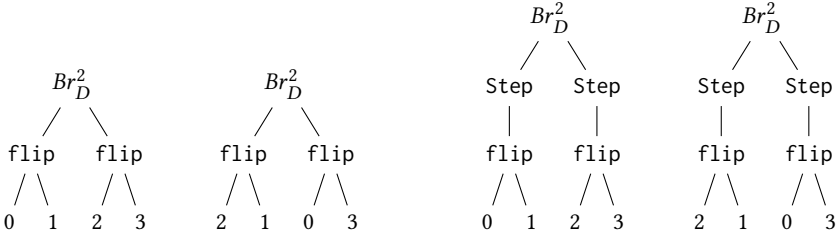
\begin{figure}
  {\small
    \hfil
    \begin{forest}
      [\brDtwo{}{}
      [\ebr [$0$] [$1$]]
      [\ebr [$2$] [$3$]]]
    \end{forest}
    \hfil
    \begin{forest}
      [\brDtwo{}{}
      [\ebr [$2$] [$1$]]
      [\ebr [$0$] [$3$]]]
    \end{forest}
    \hfil
  }
  {\small
    \hfil
    \begin{forest}
      [\brDtwo{}{}
      [\step [\ebr [$0$] [$1$]]]
      [\step [\ebr [$2$] [$3$]]]]
    \end{forest}
    \hfil
    \begin{forest}
      [\brDtwo{}{}
      [\step [\ebr [$2$] [$1$]]]
      [\step [\ebr [$0$] [$3$]]]]
    \end{forest}
    \hfil
  }
  \caption{Two strongly bisimilar trees before interpretation (left), but not after (right)}
  \label{fig:interp-counter}
\end{figure}

The counter-example in Figure~\ref{fig:interp-counter},
where $\ebr$ is the binary event introduced in Section~\ref{sec:nondeterminism},
fleshes out this intuition. Indeed, both trees are strongly bisimilar:
each of them can emit the label $\lobs{\ebr}{\false}$ by stepping to either the
$\Ret{0}$ or $\Ret{2}$ node, or emit the label
$\lobs{\ebr}{\true}$ by stepping to either the
$\Ret{1}$ or $\Ret{3}$ node.
However, they are  strongly bisimilar because the induced LTS  processes the question to the environment---$\ebr$--- and its
answer---$\false/\true$---in a single step, such that the computations never
observe that they have had access to distinct continuations.
However, if one were to introduce in the tree a \step node before the external
events, for instance using the handler \ilc{h := fun e => Step (trigger e)},
a new state allowing for witnessing the distinct continuations would become
available in the LTSs, leading to non-bisimilar interpreted trees.

We hence say a handler \ilc{h : E ~> ctree F} is \emph{simple} if it implements
each event \ilc{e} either as a pure leaf, i.e., \ilc{h e = Ret x}, or simply
translates it, i.e., \ilc{h e = trigger f}.
Similarly, we say that \ilc{h : E ~> stateT S (ctree F)} is simple if it is
point-wise simple.
We show that we recover the desired property for the subclass of simple handlers:

\begin{lemma}[Simple handlers interpret into monad morphisms
  (\linkt{Interp/Interp.v}{788})]
~\\  \vspace{-1em}
 \begin{itemize}
  \item If \ilc{h : E ~> ctree F} is simple, then $\forall t~u,~t\sbisim u \rightarrow$ \ilc{interp h} $ t \sbisim $ \ilc{interp h} $ u.$
  \item If \ilc{h : E ~> stateT (ctree F)} is simple, then $\forall
    t~u,~t\sbisim u \rightarrow \forall s,$ \ilc{interp h} $ t~s \sbisim $
    \ilc{interp h} $ u~s.$
  \end{itemize}
\end{lemma}

The latter result is, in particular, sufficient to transport \impbr equations
established before interpretation---such as the theory of \cbr{$\_$}{$\_$}---through
interpretation (\linkimp{ImpBr.v}{174}). More generally, we can reason after
interpretation to establish equations relying both on the nondeterminism and
state algebras, for instance to establish the equivalence $p_3
\equiv \cbr{p_2}{p_3}$ mentioned in Section~\ref{sec:nondeterminism} (\linkimp{ImpBr.v}{213}).

\subsection{Refinement}
\label{subsec:refinement}

Interpretation provides a general theory for the implementation of external events.
Importantly, \ctreesn also support an analogous facility for the \emph{refinement} of its
internal branches: one can shrink the set of accessible paths in a computation---and, in particular, determinize it.

\begin{figure}
  \begin{lstlisting}[style=customcoq]
    Definition refine (h : bool -> forall (n:nat), M (fin n)) : ctree E ~> M := fun R =>
    iter (fun t => match t with
    | Ret r   => ret (inr r)
    | BrD n k => bind (h false n) (fun x => ret (inl (k x)))
    | BrS n k => bind (h true  n) (fun x => ret (inl (k x)))
    | Vis e k => bind (mtrigger e) (fun x => ret (inl (k x)))
    end).
  \end{lstlisting}
  \vspace{-2ex}
  \caption{Refiner for \ctreesn (class constraints omitted) (\linkt{Interp/Refine.v}{17})}
  \label{fig:refine}
  \vspace{-2ex}
\end{figure}

We provide, to this end, a new combinator, \ilc{refine}, defined in Figure~\ref{fig:refine}.
Similar to \ilc{interp}, it takes as an argument a handler specifying how to
implement our object of interest into a monad \ilc{M}. Since here we are
interested in the implementation of branching nodes, the handler describes how
to monadically choose a branch---the boolean flag allows for different
behaviors between delayed and stepping branches.\footnote{This definition
  becomes even more uniform with \ilc{interp} once we allow for parameterized,
  arbitrary indexing of branching nodes, as  in the enhanced \ctreesn branch of our  development.}
As usual, the implementation monad must be iterative, but this time it must also
specify how to re-embed the external events using an operation called \ilc{mtrigger} (a concept already used in \citeauthor{YZZ22}~\cite{YZZ22}).

\begin{figure}
  \begin{lstlisting}[style=customcoq]
    Definition refine_cst (h : bool -> forall (n:nat), fin (S n)) : ctree E ~> ctree E :=
    refine (fun b n => match n with
    | O => stuckD
    | S n => if b then Step (Ret (h b n)) else Guard (Ret (h b n))
    end).

    Definition refine_state (h : St -> forall (b:bool) (n:nat), St * fin (S n))
    : ctree E ~> stateT St (ctree E) :=
    refine (fun b n s => match n with O => stuckD | S n => Ret (h s b n)  end).
  \end{lstlisting}
  \caption{Constant and stateful refinements (\linkt{Interp/Refine.v}{51})}
  \label{fig:refinementa}
\end{figure}

As hinted at by the combinator's name,
the source program should be able to simulate the refined program.
Fixing \ilc{M} to \ilc{ctree F}, this is expressed as
$\forall t,~$\ilc{refine h}$~t \ssim t$.
However, one cannot hope to obtain such a result for an arbitrary \ilc{h}, because it
could implement internal branches with an observable computation that can't be
simulated by $t$.
We hence provide two illustrative families of well-behaved refinements, depicted on
Figure~\ref{fig:refinementa}.
Constant refinements, implemented into \ilc{ctree E}, lift a function systematically,
returning the same branch given the nature and arity of the branching node.
Stateful refinements carry a memory into a piece of state \ilc{St} and implement
the refinement into \ilc{stateT St (ctree E)} by using the current state to pick
a branch, and update the state. A particular example of such a stateful
refinement is a round-robin scheduler.

\begin{lemma}{The constant and stateful refinements are proper refinements
    (\linkt{Interp/Refine.v}{82})
    \footnote{The second theorem actually cannot
      be stated in the main development provided as artifact: the return types of the
      trees are not identical --- the refined tree computes additionally a final
      state. We prove it in the enhanced \ctreesn branch with
      support for heterogeneous relations and arbitrary relations on labels.}
  }
  \begin{align*}
    \forall h~t,~\ilc{refine_cst}~h~t\ssim t \qquad \mathrm{and} \qquad
    \forall h~t~s,~\ilc{refine_state}~h~t~s\ssim t    
  \end{align*}  \label{lem:refine}
\end{lemma}
Finally, the shallow nature of \ctreesn also offers testing opportunities.
\citeauthor{XZHH+20}~\cite{XZHH+20} describe how external events such as IO interactions
can alternatively be implemented in \ocaml and linked against at extraction.
Similarly, we demonstrate on \impbr how to execute a \ctreen by running an impure refinement
implemented in \ocaml by picking random branches along the execution.

\subsection{ITree embedding}
\label{subsec:itree}

We have used \ctreesn directly as a domain to represent the syntax
of \impbr, as well as in our case studies (see Section \ref{sec:ccs} and \ref{sec:yield}).
\ctreesn can, however, fulfill their promise sketched in
Section~\ref{sec:background}, and be used as a domain to host the monadic implementation of
external representations of nondeterministic events in an \itreen.

\begin{figure}
  \begin{lstlisting}[style=customcoq,basicstyle=\small\ttfamily]
    Definition inject {E} : itree E ~> ctree E := interp (fun e => trigger e).
    Definition internalize {E} : ctree (Choose +' E) ~> ctree E :=
      interp (fun e => match e with | inl1 (choose n) => brs | inr1 e => trigger e).
    Definition embed {E} : itree (Choose +' E) ~> ctree E :=
      fun _ t => internalize (inject t).
    \end{lstlisting}
    \caption{Implementing external branching events into the \ctreen monad (\linkt{Interp/ITree.v}{24})}
   \label{fig:embed}
\end{figure}

To demonstrate this approach, we consider the family of events \ilc{choose (n: nat) : Choose (fin n)},
and aim to define an operator \ilc{embed} taking an \itreen computation modeling
nondeterministic branching using these events, and implementing them as stepping
branches into a \ctreen.
This operator, defined in Figure~\ref{fig:embed}, is the composition of two
transformations. First, we \ilc{inject} \itreesn into \ctreesn by (\itreen) interpretation.
This injection rebuilds the original tree as a \ctreen, where \itau nodes
have become \ilc{Guard} nodes, and an additional \ilc{Guard} has been introduced
in front of each external event.
Second, we \ilc{internalize} the external branching contained in a \ctreesn
implementing a \ilc{Choose} event, using the isomorphic stepping branch.
The resulting embedding forms a monad morphism transporting \euttn \itreesn into \sbisimn \ctreesn:
\begin{lemma}{\ilc{embed} respects \euttn (\linkt{Interp/ITree.v}{346})}
\label{lemma:embed}
\quad
\(
    \forall t~u, \euttn~ t~ u \implies embed(t)\sbisim embed(u)
\)
\end{lemma}

The proof of this theorem highlights how \itau nodes in \itreesn collapse
two distinct concepts that nondeterminism forces us to unravel in \ctreesn.
The \euttn relation is defined as the \gfpn of an inductive endofunction \tm{euttF}.
In particular, one can recursively and \emph{asymmetrically} strip finite amounts of \itau---the corecursion is completely oblivious to these nodes in the structures.
Corecursively, however, \itau nodes can be matched \emph{symmetrically}---a construction that is useful in exactly one case, namely to relate the silently spinning computation, $\itau^{\omega},$ to itself.
From the \ctreesn perspective, recursing in \tm{euttF} corresponds to recursion in
the definition of the LTS: \itau nodes are \guard nodes.
But corecursing corresponds to a step in the LTS: $\itau^{\omega}$ corresponds to $\step^{\omega}$.
\itreesn' \itau thus corresponds to either a \guard or a \step.
Nondeterminism forces us to separate both concepts, as
whether a node in the tree constitutes an accessible state in the LTS
becomes semantically relevant. We obtain two "dual"
notions that have no equivalent in \itreesn: we may have finite amounts of
asymmetric corecursive choices, i.e., finitely many \step{}s, and
structurally infinite stuck processes, i.e., $\guard^{\omega}$.

In the proof of Lemma~\ref{lemma:embed}, this materializes by the fact that an induction
on \tm{euttF} leaves us disappointed in the symmetric \itau case: we have no
applicable induction hypothesis, but expose in our embedding a $\guard$,
which does not allow us to progress in the bisimulation. We must resolve the
situation by proving that being able to step in \ilc{embed t} implies
that \ilc{t} is not $\itau^{\omega}$, i.e., that we can inductively
reach a \ilc{Vis} or a \ilc{Ret} node (\linkt{Interp/ITree.v}{276}).

\paragraph{Limitations: on guarding recursive calls using \guard.}
We have shown that \ctreesn equipped with \ilc{iter} as recursor and strong
bisimulation as equivalence form an iterative monad. Furthermore, building
interpretation atop this \ilc{iter} combinator gives rise to a
monad morphism respecting \euttn, hence is suitable
for implementing non-deterministic effects represented as external in an
\itreen.

However, we stress that the underlying design choice in the definition of
\ilc{iter}, guarding recursion using \guard, is not without consequences.
It has strong benefits, mainly that a lot of reasoning can be performed against strong
bisimulation. More specifically, this choice allows the user to reserve weak
bisimulation for the purposes of ignoring domain-specific steps of computations that may be relevant
both seen under a stepping or non-stepping lens---e.g., synchronizations
in \ccs---but it does not impose this behavior on recursive calls.
However, it also leads to a coarse-grained treatment of silent divergence:
in particular, the silently diverging \itreen (an infinite chain of \itau) is embedded into an
infinite chain of \guard, which, we have seen, corresponds to a stuck LTS.
For some applications---typically, modeling other means of being
stuck and later interpreting them into a nullary branch---one would
prefer to embed this tree into the infinite chain of \step to avoid equating
both computations. While one could rely on manually introducing a \step in the
body iterated upon when building the model, that approach is a bit cumbersome.

Instead, a valuable avenue would be to develop the theory accompanying the
alternate iterator mentioned in Section~\ref{subsec:bisim} and guarding
recursion using \step.  Naturally, the corresponding monad would not be
iterative with respect to strong bisimulation, but we conjecture that it would
be against weak bisimulation.  From this alternate iterator would arise an
alternative embedding of \itreesn into \ctreesn: we conjecture it would still
respect \euttn, but seen as a morphism into \ctreesn equipped with weak
bisimulation. The development accompanying this paper does not yet support this
alternate iterator, we leave implementing it to future work.

Currently, the user has the choice between (1) not
observing recursion at all, but getting away with strong bisimulation in
exchange, (2) manually inserting \step at recursive calls that they chose to
observe. With support for this alternate iterator, the user
would be given additional option to (3) systematically \ltau-observe recursion,
at the cost of working with weak bisimulation everywhere.



%% file: ccs.tex
We claim that \ctreesn form a versatile tool for building semantic models of
nondeterministic systems, concurrent ones in particular. In this section, we
illustrate the use of \ctreesn as a model of concurrent communicating processes
by providing a semantics for Milner's Calculus of
Communicating Systems (\ccs)~\cite{Milner:CC1989}. The results we obtain---the usual
algebra, up-to principles, and precisely the same equivalence relation as the
usual operational-based strong bisimulation---are standard, per se, but
they are all established by exploiting the generic notion of bisimilarity of \ctreesn. The result is a
shallowly embedded model for \ccs in Coq that could be easily, and modularly, combined with other
language features.

\subsection{Syntax and operational semantics}


\begin{figure}
  {
    \begin{mathpar}
      l::=~ \tau \sep \comm \sep \bar{\comm} \and
      \pP ::=~ \pnil \sep \prf{l} \pP \sep \pls\pP\pQ \sep \para\pP\pQ \sep \new\comm\pP \sep \bang{\pP}
   \end{mathpar}
  }%
\begin{mathpar}
  \inferrule{~}{\stepccs {\prf a P} a P}
  \and
  \inferrule{\stepccs P l P'}{\stepccs {\pls P Q} l P'}
  \and
  \inferrule{\stepccs Q l Q'}{\stepccs {\pls P Q} l Q'}
  \and
  \inferrule{\stepccs P l P'}{\stepccs {\para P Q} l {\para P' Q}}
  \and
  \inferrule{\stepccs Q l Q'}{\stepccs {\para P Q} l {\para P Q'}}
  \and
  \inferrule{\stepccs P \comm P' \quad \stepccs Q {\bar \comm} Q'}{\stepccs {\para P Q} \tau {\para {P'} {Q'}}}
  \and
  \inferrule{\stepccs P l P' \quad \compat l \{\comm,\bar \comm\}}{\stepccs {\new \comm P} l {\new \comm P'}}
  \and
  \inferrule{\stepccs {P \parallel \bang P} l P'}{\stepccs {\bang P} l P'}
\end{mathpar}

\caption{Syntax for \ccs (\linkccs{Syntax.v}{51}) and its operational semantics (\linkccs{Operational.v}{12})}
\label{fig:ccs}
\end{figure}

The syntax and operational semantics of \ccs are shown in Figure~\ref{fig:ccs}.
The language assumes a set of \textit{names}, or communication channels, ranged over by $\comm$.
For any name $\comm,$ there is a co-name $\bar \comm$ satisfying
$\bar {\bar \comm} = \comm$. An \textit{action} is represented
by a label $l$; it is either a communication label $\comm$ or $\bar\comm$,
representing the sending/reception of a message on a channel, or the reserved
action $\tau$, which represents an internal action.

The standard operational semantics, shown in the figure, is
expressed as a labeled transition system, where states are terms $P$ and labels
are actions $l$. The \ccs operators are the following: $\pnil$ is the process with
no behavior. A prefix process ${\prf l P}$ emits an action $l$ and then becomes
the process $P$. The internal choice $\pls\pP\pQ$ behaves either like the
process $\pP$ or like the process $\pQ$, in the same fashion as the
\textsc{BrInternal} semantics for \cbrn in Section~\ref{sec:nondeterminism}.
The parallel composition of two processes $\para\pP\pQ$ interleaves
the behavior of the two processes, while allowing the two processes to
communicate. If the process $\pP$ emits a name $\comm$ and the process $\pQ$
emits its co-name $\bar\comm$, then the two processes can progress simultaneously
and the parallel composition emits an internal action $\tau$. Channel
restriction $\new\comm\pP$ prevents the process $\pP$ from emitting an action
$\comm$ or $\bar\comm$: the operational rule states that any emission of another
action is allowed. Finally the replicated process $\bang{\pP}$ behaves as an
unbounded replication of the process $\pP$. Operationally, $\bang{\pP}$ has the
behavior of $\para P {\bang{\pP}}$.

\subsection{Model}
\begin{figure}

  \begin{mathpar}
    \spnil \defeq \stuck
    \and
    \sprf a p \defeq \triggernoparens{a}\bind p
    \and
    \spls p q \defeq \brDtwo{p}{q}
    \and \sbang p \defeq \sparabang p p
\\
    \snew{c}{P} \defeq \interp ~\tm{h\_new} ~c ~P
\and
\text{where }\tm{h\_new}~c~e =
\left\{
  \begin{array}{ll}
         {\stuck} &\text{if }e= {act ~c}  \text{ or } e=  {act ~\bar{c}} \\
         \triggernoparens{e}&\text{otherwise}
  \end{array}
\right.

    \spara{p}{q} \defeq  \cofix{F} ~p~ q \cdot
    \brDthree{}{}{}
    \begin{array}[t]{l}
    (p'\bget\gethead p\bind  \actL~\gfpc{F}~q~p') \\
    (q'\bget\gethead q\bind  \actR~\gfpc{F}~p~q')\\
    (p'\bget\gethead p\bind q'\bget\gethead q\bind actLR~\gfpc{F}~p'~q')
    \end{array}

    \\
\begin{array}{l}
  \actL~\gfpc{F}~q  ~(\brS{n}{k}) \defeq \brS{n}{(\lambda i \cdot \gfpc{F} ~(k~i)~q)}
  \\
  \actL~\gfpc{F}~q ~(\Vis{e}{k}) \defeq \Vis{e}{(\lambda i \cdot \gfpc{F} ~(k~i)~q)}
\end{array}

\and
\begin{array}{l}
\actR~\gfpc{F}~p ~(\brS{n}{k}) \defeq \brS{n}{(\lambda i \cdot \gfpc{F} ~p~(k~i))}
\\
\actR~\gfpc{F}~p ~(\Vis{e}{k}) \defeq \Vis{e}{(\lambda i \cdot \gfpc{F} ~p~(k~i))}
\end{array}
\\
  \actLR~{\gfpc{F}} ~r~r' \defeq
  \left\{
  \begin{array}{ll}
  \step \cdot \gfpc{F} ~(k~())~(k'~()) \;\;&\text{ if } \exists a.\,~ r=\Vis{(act ~a)}{k} \land r'=\Vis{(act~\bar{a})}{k'}\\
	{\stuck} &\text{ otherwise}
\end{array}\right.

 \end{mathpar}

  \caption{Denotational model for \ccs using \ccssem as a domain (\linkccs{Denotation.v}{62})}
  \label{fig:ccs-den}
\end{figure}







We define a denotational model for \ccs using \ilc{ctree actE void} as domain,
written \ccssem in the following.
As witnessed by this type, processes do not return any value, but may emit actions
modeled as external events expecting \ilc{unit} for answer: \ilc{Inductive actE
  ::= | act a : actE unit}.
Figure~\ref{fig:ccs-den} defines the semantic operators
associated with each construct of the language. They are written as over-lined
versions of their syntactic counterparts, and defined over \ccssem.


The empty process is modeled as a stuck tree---we cannot observe it.  Actions
are directly defined as visible events, and thus the prefix triggers the action,
and continues with the remaining of the process.  As discussed in
Section~\ref{sec:nondeterminism}, the delayed branching node fits exactly with
the semantics of the choice operator in \ccs, only progressing if one of the
composed terms progresses.  Restriction raises a minor issue: the compositional
definition implies that the \ctreen for the restricted term has already been
produced when we encounter the restriction and, a priori, that tree might
contain visible actions on the name being restricted. We enforce scoping by
replacing those actions by a stuck tree, $\stuck$, effectively cutting these
branches.  This is done using the \interp\ operator from \ctreesn, with
\tm{h\_new}, a handler that does the substitution.

Parallel composition is more intricate, as the operator
requires significant introspection of the composed terms. The traditional
operational semantics of \ccs is not explicitly constructive:
each of the three reduction rules depends on the existence of specific
transitions in the sub-processes.
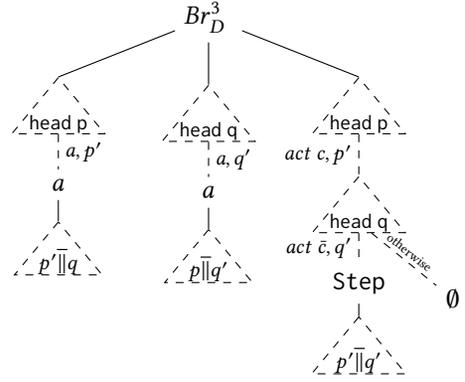
\begin{wrapfigure}{r}{0.5\textwidth}
  \centering
  \input{ccs-schema}
  \caption{Depiction of the tree resulting from $\spara{p}{q}$}
  \label{fig:ccs-spara-tree}
  \vspace{-2ex}
\end{wrapfigure}
We perform this necessary introspection, in a constructive way, defining the
tree as an explicit cofixpoint, and using the \gethead\ operator introduced in
Section~\ref{subsec:head}.  While the operation \gethead{p} precisely captures
the desired set of actions that $p$ may perform, computing this set could, in
general, silently diverge. We therefore cannot bluntly initiate the computation
by sequencing the heads of $p$ and $q$, as divergence in the former may render
inaccessible valid transitions in the latter.\footnote{Technically, the variant
  of \ccs considered here actually cannot generate such a computation, so we
  could therefore rule this case out extensionally. The case could, however,
  easily arise in a variant of \ccs relying on recursive processes rather than
  replication, and in other calculi, so we therefore favor this more general,
  reusable, approach.}  Instead, we initiate the computation with a ternary
delayed choice: the left (resp.  middle) branch captures the behaviors starting
with an interaction by $p$ (resp. $q$), while the right branch captures the
behaviors starting with a synchronisation between $p$ and $q$. Essentially, the
tree nondeterministically explores the set of applicable instances of the three
operational rules for parallel composition.  In particular,
if the operational rule to step in the left (resp. right) process is
non-applicable, the left (resp. middle) branch of the resulting tree silently
diverges.  The right branch silently diverges if neither process can step, but,
in general, it also contains branches considering the interaction of
incompatible actions; we cut these branches by inserting \stuck.  In all cases,
the operator continues corecursively, having progressed in either or both
processes.  Figure~\ref{fig:ccs-spara-tree} shows the \ctreen resulting from $\spara{p}{q}$.


The last operator to consider is the replication $\sbang$. In theory, it
could be expressed in terms of parallel composition directly, as the cofix
$\sbang p \defeq \spara{p}{~\sbang{p}}$. Unfortunately, although it is sound, defining the
$\sbang{}$ operator in this way is too involved for \coq's syntactic criterion on
\cofixes to recognize that the corecursive call is guarded under $\spara{}{}$.
To circumvent this difficulty, we use an auxiliary operator,
$\sparabang{p}{q}$,  capturing the parallel composition of a process
$p$ with a replicated process $\bang{q}$.
It nondeterministically explores
the four kinds of interactions that such a process could exhibit: a step in $p$;
the creation of a copy of $q$ performing a step to $p'$, before being composed
in parallel with $p$; the creation of a copy of $q$ synchronizing with $p$; or
the emission of two copies of $q$ synchronizing one with another.  We finally
define the replication operator as $\sbang p \defeq \sparabang{p}{p}$. We omit the formal definition here,
since it is similar to that for $\spara{}{}$, but point out the interested
reader to its formal counterpart (\linkccs{Denotation.v}{125})\arxiv{ or to the appendix}{}.

With these tools at hand, the model $\sembra{\cdot}:~ \ccs \rightarrow \ccssem$
is defined by recursion on the syntax.

\paragraph{Equational Theory}

We provide a first validation of our model by proving that it satisfies the expected equational
theory with respect to \ctreesn's notion of strong bisimulation, enabling the usual
algebraic reasoning advocated for process calculi.
In particular, we prove that our  definition for the replication is
sane in that it validates equationally the expected definition: $\sbang{p}
\sbisim \spara {\sbang{p}}{p}$ (\linkccs{Denotation.v}{1452}).
We also prove an illustrative collection of expected equations satisfied
by our operators (\linkccs{Denotation.v}{809}):
\begin{mathpar}
  \spls p q \sbisim \spls q p
  \and
  \spls p {(\spls q r)} \sbisim \spls {(\spls p q)} r
  \and
  \spls p \spnil \sbisim p
  \and
  \spls p p \sbisim p
  \\
  \spara p \spnil \sbisim p
  \and
  \spara p q \sbisim \spara q p
  \and
  \spara p {(\spara q r)} \sbisim \spara {(\spara p q)} r
\end{mathpar}

To facilitate these proofs, we first prove sound up-to principles at the level
of \ccs for each constructor: strong bisimulation up-to $\sprf{c}{[\cdot]}$,
$\spls{[\cdot]}{[\cdot]}$, $\spara{[\cdot]}{[\cdot]},$ $\sbang{[\cdot]},$
and $\snew{c}{[\cdot]}$ are all valid principles, allowing us to rewrite \sbisimn under
semantic contexts during bisimulation proofs. Additionally to these language-level up-to
principles, we  inherit the ones generically supported by \sbisimn
~(Lem.~\ref{lem:sbisim-upto}).







\subsection{Equivalence with the operational strong bisimilarity}

In addition to proving that we recover in our semantic domain the expected up-to
principles and the right algebra, we furthermore show that the model is sound
and complete with respect to strong bisimulation compared to its
operational counterpart.
We do so by first establishing an asymmetrical bisimulation between \ccs and
\ccssem, matching operational steps over the syntax to semantic steps in the
\ctreen. We write $\bar l$ for the obvious translation of labels between both LTSs.

\begin{definition}[Strong bisimulation between \ccs and \ccssem.]
A relation $\mathcal{R}:~ rel(\ccs,\ccssem)$ is a strong bisimulation if and only if,
for any label $l$, \ccs term $P$, and  \ccssem tree $q$.\\
\begin{minipage}[c]{.7\textwidth}
\[
P ~ \mathcal{R} ~ q
\land \stepccs{P}{l}{P'}
\implies
\exists q',~ P' ~\mathcal{R}~ q'
\land \lstep{q}{\bar{l}}{q'}
\]
\begin{center}
and conversely
\end{center}
\vspace{-1ex}
\[
P ~ \mathcal{R} ~ q
\land \lstep{q}{\bar{l}}{q'}
\implies
\exists P'.~ P' ~\mathcal{R}~ q'
\land \lstep{P}{l}{P'}
\]
\vspace{.5ex}
\end{minipage}\begin{minipage}[c]{.25\textwidth}
\begin{tikzpicture}
	\node (s) at (-1,1) {$P$};
	\node (R) at (0,1) {$\mathcal{R}$};
	\node (t) at (1,1) {$q$};
	\node (sp) at (-1,0) {$P^\prime$}
		edge [<-, shorten <=.06cm] node[auto] {$l$} (s);
	\node (Rp) at (0,0) {$\mathcal{R}$};
	\node (tp) at (1,0) {$q^\prime$}
  edge [<-, shorten <=.04cm] node[auto,swap] {$\bar{l}$} (t);
	\node (ccs) at (-0.84,0.27) {\tiny{ccs}};
\end{tikzpicture}
\end{minipage}
\end{definition}

\begin{lemma}[\linkccs{OpDenot.v}{251}]
  The relation $R \defeq \{(P,q)~\mid~\sembra{P}\sbisim q\}$ is a strong bisimulation.
\end{lemma}

We derive from this result that the operational and semantic strong
bisimulations define exactly the same relation over \ccs:

\begin{lemma}[\linkccs{OpDenot.v}{430}]
  \( \forall P~Q,~\sembra{P}\sbisim\sembra{Q} \text{ iff } P \sim_{\ccs} Q\)
\end{lemma}





%% file: ccs-schema.tex
\usetikzlibrary{trees}

\tikzset{
  subtree/.style={
    draw,dashed,
    isosceles triangle,shape border rotate=90,
    isosceles triangle apex angle=80,
    inner sep=0.5pt,
    anchor=apex,
    scale=0.8,
    minimum size=20pt,
  },
  edge/.style={
    dashed,
  },
}
  
\begin{tikzpicture}[on grid, node distance=1.4cm and 2cm]
  \node (root) {$\brDthree{}{}{}$};

  \node[subtree, below left=of root] (hpl) { $\gethead p$};
  \node[below=0.8cm of hpl] (a1) {$a$};
  \node[subtree, below=1cm of a1] (rec1) {$\spara{p'}{q}$ };
  \draw (root) to (hpl.apex);
  \draw[edge] (hpl) to node[scale=0.8,right]{$a, p'$} (a1);
  \draw (a1) to (rec1);
  
  \node[subtree,below=1.5cm of root] (hql) { $\gethead q$};
  \node[below=0.8cm of hql] (a2) {$a$};
  \node[subtree, below=1cm of a2] (rec2) {$\spara{p}{q'}$ };
  \draw (root) to (hql.apex);
  \draw[edge] (hql) to node[scale=0.8,right]{$a, q'$} (a2);
  \draw (a2) to (rec2);
  
  \node[subtree,below right=of root] (hpr) { $\gethead p$};
  \node[subtree,below =1.3cm of hpr] (hqr) { $\gethead q$};
  \node[below=0.8cm of hqr] (step) {$\step$};
  \node[subtree, below=1cm of step] (rec3) {$\spara{p'}{q'}$ };
  \draw (root) to (hpr.apex);
  \draw[edge] (hpr) to node[scale=0.8,left]{$act~c, p'$} (hqr.apex);
  \draw[edge] (hqr) to node[scale=0.8,left]{$act~\bar{c}, q'$} (step);
  \draw (step) to (rec3);

  \node[below right=1.00cm and 1.25cm of hqr] (stuck) {$\stuck$};
  \draw[edge] (hqr) to node[scale=0.5, above, sloped]{otherwise} (stuck);
\end{tikzpicture}

%% file: yield.tex
As a second case study, we consider cooperative multithreading.
Cooperative scheduling is used in languages such as Javascript, async
Rust, or Akka/Scala actors, but is also a very general model, as preemptive
multi-tasking is equivalent to cooperative scheduling where threads are willing
to yield (i.e., let other threads run) at any time.
We extend the syntax of \imp with two constructs:
%
\begin{mathpar}
  \ccomm \defeq \cskip \sep \cassign{x}{e} \sep \cseq{c1}{c2} \sep \cwhile{b}{c} \sep \cfork{c1}{c2} \sep \cyield
\end{mathpar}

The command \efork forks two copies of the currently running program, the first
of which first runs $c1$ and the second of which first runs $c2$,
before they each proceed with the rest of the program.
Importantly, $c2$ retains control, and $c1$ will only execute when $c2$ voluntarily yields control or ends.
This yielding of control is achieved by the $\cyield$ statement, which signals
that the current thread stops and lets a new thread run---possibly the same one
again.

This semantics implements a mechanism like the \tm{fork} system call, which duplicates the current process, but without a possibility for joining threads.
For example, the program
\[
\cseq
    {(\cfork
      {(\cassign{x}{1})}
      {(\cseq
        {\cyield}
        {\cassign{x}{2}}
        )}
      )}
    {\cassign{y}{x}}
\]
forks two copies of the program, with the ``main'' thread immediately yielding, allowing for either thread to run next.
Its semantics is to first spawn a thread for $\cassign{x}{1}$, then have the
main thread reach the \cyield, giving $\cassign{x}{1}$ a chance to run.
Assuming the spawned thread goes next, it runs in sequence $\cassign{x}{1}$ and
$\cassign{y}{x}$,
after which the main thread recovers control and finishes its execution.
$\cassign{y}{x}$ is part of both threads and is thus executed twice, after each assignment to $x$.

Alternatively, some cooperative scheduling languages~\cite{Abadi2010} consider a spawn
operator that simply spawns an independent thread: we can encode this behavior
in several ways.
Notably, if $\tm{spawn}$ always occurs in tail position, there is no
continuation to duplicate. For instance, the program
\[
\cfork{\cassign{x}{1}}
      {(\cfork
        {(\cassign{x}{2})}
        {\cskip}
        )}
\]
spawns two threads that set $x$ to different values, and terminates.
The two spawned threads can then be scheduled in either order, resulting
in $x = 1$ or $x = 2$ in the final state.
This constraint could be syntactically enforced in the language if relevant.
Alternatively, one can use the command $\cwhile{\true}{\cyield}$ to
``terminate'' a thread; for instance to prevent the first thread above from
reaching $\cassign{y}{x}$. With fancier encodings using reserved
shared-variables, nested ``joins'' and other synchronization operations can be
modeled.
In this case study, we do not concern ourselves with such extensions, and
restrict ourselves to the formalization of the syntax described above.


\subsection{Model}
\label{sec:yieldmodel}

As with the approach described in Section~\ref{sec:background}, the semantics of
the language is defined in two stages.
First, we represent statements as computations of type \ilc{ctree (YieldE +
  SpawnE + MemE) unit}, where \ilc{YieldE} and \ilc{SpawnE}
are used to represent \cyield and \efork respectively.
The remaining statements of \imp are defined in a standard way---we omit them.
The new event families are defined as follows:

\begin{minipage}[b]{2.5in}
\begin{lstlisting}[style=customcoq]
  Variant YieldE : Type -> Type :=
  | Yield : YieldE unit.
\end{lstlisting}
\end{minipage}
\begin{minipage}[b]{2.5in}
\begin{lstlisting}[style=customcoq]
  Variant SpawnE : Type -> Type :=
  | Spawn : SpawnE bool.
\end{lstlisting}
\end{minipage}

\ilc{Yield} carries no additional information and acts purely as a signal to yield control, and \ilc{Spawn} introduces a binary branch in the \ctreen, allowing us to store the asynchronous thread in one branch and the main thread that continues running in the other branch.
Formally, these events are used to denote the corresponding statements:
\begin{align*}
    \sembra{\cyield} \defeq& \trigger{\ilc{Yield}} \qquad
    \sembra{\cfork{c1}{c2}} \defeq& b \gets \trigger{\ilc{Spawn}} \bind \text{if } b \text{ then } \sembra{c1} \text{ else } \sembra{c2}
\end{align*}

The second stage is more complex than the simple stateful
interpretation of standard \imp: we need to interpret the concurrency-related events.
In particular, since the semantics of these events involves scheduling other threads that work
together, we cannot hope to use \interp, as is done with memory events---this
situation is similar to $\spara{}{}$ for \ccs, in both cases the combinator
cannot be implemented as a simple fold.
Instead, we define a custom interpreter for \ilc{Spawn} and \ilc{Yield} events,
namely a scheduler operating over a thread pool of threads, each of which is a
\ctreen that may contain \ilc{Spawn} and \ilc{Yield} events,
and producing a final \ctreen devoid of these concurrency events.

\newcommand{\SOME}[1]{\ensuremath{\lfloor #1 \rfloor}\xspace}
\newcommand{\NONE}{\SOME{\mbox{-}}\xspace}

\begin{figure}
  \[
\begin{array}{rcl}
\schedule~v_0~\NONE &\defeq& \ret \unit \\
\schedule~v_n~\NONE &\defeq& \brS{n}{(\lambda i \cdot \schedule~v_n~\SOME{i})} \qquad \text{for } n > 0\\
\schedule~v_n~\SOME{i} &\defeq \\
\mathrm{if}\ v[i] = \quad& & \mathrm{then} \\
\ret \unit & & \brD{1}{(\lambda \_ \cdot \schedule~v[-i]_{n-1}~\NONE)} \\
\brS{m}{k} & & \brS{m}{(\lambda i' \cdot \schedule~v[i \mapsto k~i']_n~\SOME{i})} \\
\brD{m}{k} & & \brD{m}{(\lambda i' \cdot \schedule~v[i \mapsto k~i']_n~\SOME{i})} \\
\Vis{\ilc{Yield}}{k} & & \brD{1}{(\lambda \_ \cdot \schedule~v[i \mapsto k~\unit]_n~\NONE)} \\
\Vis{\ilc{Spawn}}{k} & & \brS{1}{(\lambda \_ \cdot \schedule~(k~\true :: v[i \mapsto k~\false])_{n+1}~\SOME{i + 1})} \\
\Vis{\ilc{e}}{k} & & \Vis{e}{(\lambda x \cdot \schedule~(v[i \mapsto k~x])_n~\SOME{i})}
\end{array}
\]
\caption{The definition of \schedule (\linkyield{Par.v}{626})}
\label{fig:schedule}
\end{figure}

As is usual with cooperative scheduling, the thread pool, encoded as a vector, contains a designated
\emph{active} thread, i.e., the currently running thread.
We write $v_n$ for a vector of size $n$, and vector operations for removing the
$i$-th element as $v[-i]$, updating the $i$-th element to $x$ as $v[i \mapsto
x]$, and adding an element $x$ to the front as $x :: v$
(so $x$ is the new $0$-th element in the resulting vector).

The scheduler, $\schedule$, is defined formally in Figure~\ref{fig:schedule}.
References to $\schedule$ in its body should be interpreted as corecursive
calls---we abuse notations to lighten the presentation.
It takes two arguments: the thread pool, and the index of the active
thread in the pool. If no thread is active, the second argument is empty, written \NONE (we use an option datatype in \coq).
If there is an active thread, \SOME{i}, the $\schedule$ makes progress in that
thread.
Otherwise, $\schedule$ nondeterministically chooses the next thread to run.
If the thread pool is empty, then $\schedule$ emits a returning \ctreen.
In more detail, there are several ways the active thread may progress.
If the thread performs a branching or a memory operation, the $\schedule$
returns a \ctreen that also performs the same operation,
and in each branch replaces the active thread (index $i$)
in the thread pool with its continuation.
If the active thread returns, it is removed from the thread pool; there is no active thread after the removal.
If the thread yields, then $\schedule$ removes the \ilc{Yield} and starts
scheduling with no active thread (which will nondeterminstically pick a new active thread, potentially the same one again).
If the thread spawns a new thread, then $\schedule$ adds the $\true$ branch of
that thread to the thread pool, updates the active thread to the $\false$
branch, and continues running on the $\false$ branch (note that this shifts the active thread from index $i$ to $i+1$).

Using $\schedule$, we can interpret the $\ilc{Yield}$ and $\ilc{Spawn}$ events in concert, allowing us to eliminate them from the \ctreesn.
We denote this ``scheduled'' semantics by  $\ssembra{p} \defeq
\schedule~[\sembra{p}]_1~\NONE: \ilc{ctree MemE unit}$: it denotes a \ctreen with only memory events.
As described in Section~\ref{sec:itrees},
the final step of the second stage of modeling our \imp programs  can use a handler for the memory events to obtain a computation in the state monad of type \ilc{stateT mem (ctree voidE) unit}.

\subsection{Equational theory}

The model described in Section~\ref{sec:yieldmodel} allows us to derive some
program equivalences at source-level w.r.t. weak bisimilarity of their models.
For example, the following programs all just run $c$, though some of them first perform
some ``invisible'' steps related to concurrency (\linkyield{Lang.v}{331}):
\[
\ssembra{\cfork {c} {\cskip}} \wbisim
\ssembra{\cseq{\cyield}{c}} \wbisim
\ssembra{c}
\]
We use weak rather than strong bisimilarity
since $\schedule$ can introduce stepping branches when interpreting the two concurrency events.
We emphasize that these equations are not compositional, they only hold in the absence of additional concurrent threads, hence why \cyield behaves as a no op: the monadic equivalence we obtain is a congruence for the source language after the first stage of interpretation, but not after the second stage.

Other equations, especially ones that make use of multiple threads in nontrivial
ways, rely on the stability of $\schedule$ under \sbisimn:
\begin{lemma}[$\schedule$ preserves $\sbisim$ (\linkyield{Par.v}{1908})]
  \label{lem:schedule}
  If the delayed branches of every element of vectors $v_n$ and $w_n$ have arity less than 2, and the elements of both vectors are strongly bisimilar up to a permutation $\rho$, then $\schedule~v_n~\SOME{i} \sbisim \schedule~w_n~\SOME{\rho~i}$.
\end{lemma}
The arity requirement is satisfied by all denotations of programs in this language.
This condition greatly simplifies the proof by constraining the shape that the
strongly bisimilar \ctreesn can take.

Lemma~\ref{lem:schedule} allows us to permute the thread pool, which is useful
in examples such as (\linkyield{Lang.v}{247}):
\[
\ssembra{\cfork {c1} {(\cfork {c2} {\cskip})}} \wbisim
\ssembra{\cfork {c2} {(\cfork {c1} {\cskip})}}
\]
This program spawns two asynchronous threads then yields control to one of the two.
This equivalence captures the natural fact that it does not matter which thread is spawned first, since neither can run until both are spawned.

Furthermore, Lemma~\ref{lem:schedule} allows us to validate some simple optimizations that do not directly involve reasoning about concurrency or memory, such as (\linkyield{Lang.v}{363}):
\begin{align*}
  &\hfil\ssembra{\cfork {c1} {(\cfork {(\cwhile{\true}{\cyield})} {\cskip})}}\hfil\\
  \wbisim\quad&\hfil\ssembra{\cfork {(\cseq{\cyield}{\cwhile{\true}{\cyield})} {(\cfork {c1} {\cskip})}}}\hfil
\end{align*}
This example is similar to the previous one, except that one of the spawned threads is a while loop, which we wish to unroll by one iteration.
Crucially, the loop and its unrolled form are strongly bisimilar, so this equivalence follows from Lemma~\ref{lem:schedule} just as in the previous example.
Other optimizations that can be done before interpreting events, such as
constant folding or dead code elimination, can be proven sound similarly.

Finally, equivalences involving memory operations are still valid as well (\linkyield{Lang.v}{443}):
\[
\cfork {(\cassign{x}{2})} {(\cassign{x}{1})} \equiv
\cassign{x}{2}
\]
where $\equiv$ here refers to equivalence ($\wbisim$ in this case) after interpreting both concurrency and memory events.
This result follows from the result in Section~\ref{subsec:interp}, which allows us to transport equations made before interpreting state events into computations in the state monad after interpetation.


%% file: rw.tex
Since Milner's seminal work on \ccs~\cite{Milner:CC1989} and the $\pi$-calculus~\cite{picalculus},
process algebras have been the topic of a vast literature~\cite{BPS01}.
We mention only a few parts of it that are most relevant to our work.
In the Coq realm, Lenglet et al.~\cite{lengletS18} have formalized HO$\pi$,
a minimal $\pi$-calculus, notably exhibiting the difficulty inherent to the
formal treatment of name extrusion.
Beyond its formalization, dealing with scope extrusion as part of a compositional semantics
is known to be a challenging problem~\cite{Crafa2012,Cristescu2013}.
By restricting to \ccs in our case study, we have side-stepped this difficulty.
Foster et al.~\cite{FHW21} formalize in Isabelle/HOL a semantics for CSP
and the Circus language using a variant implementation of \itreesn, where
continuations to external events are partial functions. They, however, only model
deterministic processes, leaving support for nondeterministic ones to future
work. This paper introduces the tools to address that problem.
CSP has also been extensively studied by Brookes~\cite{Brookes02} by providing a
model based on the compositional construction of infinite sets of traces:
\ctreesn offers a complementary coinductive model to this more set-theoretic approach.
Furthermore, Brookes tackles questions of fairness, an avenue that we have not yet
explored in our setup.

Formal semantics for nondeterminism are especially relevant when dealing with
low-level concurrent semantics.  In shared-memory-based programming languages,
rather than message passing ones, concurrency gives rise to the additional
challenge of modeling their memory models, a topic that has received considerable
attention.
Understanding whether monadic approaches such as the one proposed in this paper
are viable to tackle such models vastly remains to be investigated.
Early suggestions that they may include \citeauthor{LesaniXKBCPZ22}~\cite{LesaniXKBCPZ22}
: the authors prove correct concurrent objects
implemented using \itreesn, assuming a sequentially consistent model of shared
memory.  They relate the \itreesn semantics to a trace-based one to reason about
refinement, something that we conjecture would not be necessary when starting
from \ctreesn.
Operationally specified memory models, in the style of which
increasingly relaxed models have been captured and sometimes formalized,
intuitively seem to be a better fit.
Major landmarks in this axis include the work by Sevcík
et al. on modeling TSO using a central synchronizing transition system linking
the program semantics to the memory model in the CompCertTSO
compiler~\cite{compcerttso}; or Kang et al.'s promising
semantics~\cite{promising,promising2} that have captured large subsets of the C++11
concurrency model without introducing out-of-thin-air behaviors.
On the other side of the spectrum,
axiomatic models in the style of Alglave et al.'s~\cite{cat1,cat2}
framework appear less likely to transpose to our constructive setup.

Our model for cooperative multithreading is partially reminiscent of Abadi and Plotkin's work~\cite{Abadi2010}:
they define a denotational semantics based on partial traces that they prove fully abstract, and satisfying
an algebra of stateful concurrency.
The main difference between the two approaches is that partial traces use the memory state explicitly to define the
composition of traces, where \ctreesn can express the semantics of a similar language independently of the memory model.
The formal model we describe here tackles a slightly different language than theirs, but we should
be able to adapt it reasonably easily to obtain a formalization of their work.
More recently, Din et al.~\cite{2017-LAGC-Old,2022-LAGC-New} have suggested
a novel way to define semantics based on the composition of symbolic traces,
partially inspired by symbolic execution~\cite{King76}.
They use it, in particular, to formalize actor languages, which rely on cooperative scheduling,
with a similar modularity as the one we achieve (orthogonal semantic features can be composed),
but not in a compositional way.

Our work brings proper support for nondeterminism to monadic interpreters in
\coq. As with \itreesn however, the tools we provide are just right to
conveniently build denotational models of first order languages, such as \ccs,
but have difficulty retaining compositionality when dealing with higher-order languages.
In contrast, on paper, game semantics has brought a variety of techniques lifting this limitation.
In particular in a concurrent setup, event structures have spawned a successful line of
work~\cite{RW11,CastellanCRW17} from which inspiration could be drawn for further work on \ctreesn.

%% file: conclusion.tex
We have introduced \ctreesn, a model for nondeterministic,
recursive, and impure programs in \coq. Inspired by \itreesn, we have introduced
two kinds of nondeterministic branching nodes, and designed a toolbox to reason
about these new computations.
In addition to the \coq library presented in this paper, we have developed an enhanced
implementation (with heavier notations, but more expressive) supporting arbitrary
indexing of branches and heterogeneous bisimulations between \ctreesn.
We have illustrated the expressiveness of the framework through two significant case studies.
Both nonetheless offer avenues for further work, notably through an extension of \ccs to
name passing \emph{à la} $\pi$-calculus, and to further extend the equational theory
for cooperative multithreading that we currently support.
 
The restriction imposed on handlers in Section~\ref{subsec:interp} suggests
another perspective: rather than requiring handlers to implement events in a
non-observable way, we could prevent the programs in Figure~\ref{fig:interp-counter}
from being bisimilar in the first place.
We conjecture that this could be achieved by tweaking the LTS so that
external events induce two steps: the question to the environment first, and the answer second,
resulting in a finer grained bisimilarity. 
Whether such a finer bisimilarity would still allow us to equate all
the programs we want needs further investigation.


%% file: appendix.tex
\section{Elimination of stepping branches: a proof sketch}

\begin{lemma}{Stepping nondeterministic branches can be eliminated (\linkt{Misc/brSElim.v}{249})}
  \[\forall t,~\ilc{BrSElim}~t\sbisim t\]
\end{lemma}

\begin{proof}
We proceed by coinduction, with the natural candidate: $\{(\ilc{BrSElim}~ t, t)~\mid~t\ilc{ : ctree E X}\}$.
We then case-analyse on the structure of $t$, and
partially reduce a step of iteration in each case (formally, this reduction
relies on the eta-unfolding of \ilc{unary_stepping_form}, which only holds
with respect to \equn, and hence is made possible by the validity of strong bisimulation up-to \equn). 
\begin{itemize}
  \item Case \Ret{x}: we can simply conclude by reflexivity (validity of \reflC).
  \item Case \Vis{e}{k}: the left-hand-side is of the shape
    $\triggernoparens{e}~\bindc \lambda x \Rightarrow \guard~(\ilc{BrSElim}~(k~x))$.
    It rewrites to
    $\Vis{e}{\lambda x \Rightarrow \guard~(\ilc{BrSElim}~(k~x))}$
    (\uptorelC{\equn}),
    from which we can use the \visn structural rule
    (Figure~\ref{fig:sbisim-upto})
    to step without having to explicitly consider each process stepping.
    At this point, we can almost conclude by coinductive hypothesis, except for
    the trailing \guard on the left-hand side; however, we can simply eliminate
    it by rewriting the equation $\guard~t \sbisim t$ (\uptorelC{\sbisimn}).
  \item Case \brS{n}{k}: We need to play explicitly the game in both directions---taking the branch $x$ of the \brS{n}{} node is the same as taking both $x$
    in \brD{n}{} and crossing the \step.
    We then conclude by coinduction by simply swiping off a $\guard$ as
    in the previous case (\uptorelC{\sbisimn}). 
  \item Case \brD{n}{k}: Things are more involved in this case. Indeed, the case
    analysis on $t$ has exposed a bit of syntax, a \brD{n}{} node, that appears
    on both side of the equation. However, as highlighted in
    Figure~\ref{fig:sbisim-upto}, processing these nodes does \emph{not} 
    give access to the coinduction hypothesis.
    Rather, we introduce auxiliary lemmas to capture the state of the trees that
    correspond to the next state of the LTS: one to transport a step in
    \ilc{BrSElim} to one in its argument, and one the other way around.
    \begin{mathpar}
      \small{
      \inferrule* [right=L1]
      {\lstep{\mathtt{BrSElim}~t}{l}{u}}
      {\exists t',~\lstep{t}{l}{t'}\land u \sbisim \mathtt{BrSElim}~t'}

      \inferrule*[Right=L2]{\lstep{t}{l}{u}}{\exists
        u',~\lstep{\mathtt{BrSElim}~t}{l}{u'}\land u' \sbisim
        \mathtt{BrSElim}~u}
      }
    \end{mathpar}
    Each of these lemmas are proved by reinforcing them to structurally specify
    the shape of the arguments with respect to \equn, and then proceeding by induction on
    the transition.
    Once we have these lemmas at hand, the main proof is straightforward: we
    transport a step in the function to its argument by L1, use the hypothesis of
    bisimilarity to transport it to the argument on the other side of the game,
    and transport it back to the function by L2, allowing us to
    conclude by coinduction hypothesis. 
\end{itemize}
\end{proof}

\section{A semantic replication operator in CCS$^{\#}$}

For completeness, we provide below the denotational model of the $\sbang$ operator. The principles of this definition has been provided in the body of the article. Recall that we use an auxiliary operator $\sparabang {} {}$ to ensure syntactically that the coinductive definition is well-guarded.

  \begin{mathpar}

 \sbang p \defeq \sparabang p p    \\\\

 \sparabang p q \defeq  \cofix{F} ~p~ q  \cdot
 \brDfour ~~
 \begin{array}[t]{l}
 (p'\bget\gethead p\bind  \actL~ \gfpc{F}~q~p' ) \\
 (q'\bget\gethead q\bind  \actRPB~ \gfpc{F}~p~q')\\
 (p'\bget\gethead p\bind q'\bget\gethead{q}\bind \actLRPB~\gfpc{F}~p'~q')\\
 (q'\bget\gethead q\bind q''\bget\gethead{q}\bind \actRRPB~\gfpc{F}~q'~q'')
 \end{array}

\begin{array}{l}
  \actL~\gfpc{F}~q  ~(\brS{n}{k}) \defeq \brS{n}{(\lambda i \cdot \gfpc{F} ~(k~i)~q)}
  \\
  \actL~\gfpc{F}~q ~(\Vis{e}{k}) \defeq \Vis{e}{(\lambda i \cdot \gfpc{F} ~(k~i)~q)}
\end{array}

 \begin{array}{l}
 \actRPB~\gfpc{F}~p ~(\brS{n}{k}) \defeq \brS{n}{(\lambda i \cdot \gfpc{F} ~(\spara p {k~i})~q )}
 \\
 \actRPB~\gfpc{F}~p ~(\Vis{e}{k}) \defeq \Vis{e}{(\lambda i \cdot \gfpc{F} ~(\spara p {k~i})~q)}
 \end{array}

 \actLRPB~{\gfpc{F}} ~r~r' \defeq 
 \left\{
 \begin{array}{ll}					
 	\step \cdot \gfpc{F} ~(\spara {k~()}{k'~()})~q \;\;\;\;&\text{ if } \exists a.\, r=\Vis{(act ~a)}{k} \land r'=\Vis{(act~\bar{a})}{k'}\\
 	{\stuck} &\text{otherwise} 
 \end{array}
 \right.
 \\

 \actRRPB~{\gfpc{F}} ~r~r' \defeq 
 \left\{
 \begin{array}{ll}					
 	\step \cdot \gfpc{F} ~(\spara p {\spara {k~()}{k'~()}})~q &\text{ if } \exists a.\,~ r=\Vis{(act ~a)}{k}\land r'=\Vis{(act~\bar{a})}{k'}\\
 	{\stuck} &\text{otherwise} 
 \end{array}
 \right.
 \end{mathpar}
